\documentclass[10pt, twocolumn, journal]{IEEEtran}

\usepackage{amsmath, amsthm, amssymb, mathtools, bm}  
\usepackage{graphicx}
\usepackage{stfloats}      
\usepackage{array}
\usepackage{multirow}
\usepackage[font=small,labelfont=bf]{caption}
\usepackage{subcaption}
\usepackage[inline,shortlabels]{enumitem}  
\usepackage{microtype}  
\usepackage[colorlinks=true, allcolors=blue]{hyperref}
\usepackage{orcidlink}
\usepackage{pgfplots}
\pgfplotsset{compat=1.17}

\mathtoolsset{mathic=true}    
\usepackage[ruled,vlined]{algorithm2e}
\SetAlFnt{\small}
\SetAlCapFnt{\small}
\SetAlCapNameFnt{\small}
\SetAlgoLined
\DontPrintSemicolon

\newtheorem{theorem}{Theorem}
\newtheorem{lemma}{Lemma}
\newtheorem{proposition}{Proposition}
\theoremstyle{definition}

\newtheorem{corollary}{Corollary}

\DeclareMathOperator{\logit}{logit}

\newcommand{\sig}{\operatorname{sig}}  

\newcommand{\EX}{\mathbb{E}}

\newcommand{\J}{\mathrm{j}}

\setlist{noitemsep,leftmargin=*,topsep=2pt,partopsep=1pt}
\usepackage{titlesec}
\titlespacing*{\section}{0pt}{0.6ex}{0.4ex}
\titlespacing*{\subsection}{0pt}{0.4ex}{0.2ex}
\titlespacing*{\subsubsection}{0pt}{0.3ex}{0.1ex}

\title{Adaptive Learning for IRS-Assisted Wireless Networks: Securing Opportunistic Communications Against Byzantine Eavesdroppers}

\author{%
  Amirhossein Taherpour~\orcidlink{0000-0003-4647-102X},%
  \thanks{A. Taherpour is with the Department of Electrical Engineering, Columbia University, New York, NY, USA (e-mails: at3532@columbia.edu).}%
  Abbas Taherpour~\orcidlink{0000-0003-0706-5774},~\IEEEmembership{Senior Member,~IEEE},%
  \thanks{A. Taherpour is with the Department of Electrical Engineering, Imam Khomeini International University, Qazvin, Iran (e-mail: taherpour@ikiu.ac.ir).}%
  Tamer Khattab~\orcidlink{0000-0003-2347-9555},~\IEEEmembership{Senior Member,~IEEE}%
  \thanks{T. Khattab is with the Department of Electrical Engineering, Qatar University, Doha, Qatar (e-mail: tkhattab@ieee.org).}%
}

\date{}

\begin{document}

\maketitle
\thispagestyle{empty}

\begin{abstract}
We propose a joint learning framework for Byzantine-resilient spectrum sensing and secure intelligent reflecting surface (IRS)–assisted opportunistic access under channel state information (CSI) uncertainty. The sensing stage performs logit-domain Bayesian updates with trimmed aggregation and attention-weighted consensus, and the base station (BS) fuses network beliefs with a conservative minimum rule, preserving detection accuracy under a bounded number of Byzantine users. Conditioned on the sensing outcome, we pose downlink design as sum mean-squared error (MSE) minimization under transmit-power and signal-leakage constraints and jointly optimize the BS precoder, IRS phase shifts, and user equalizers. With partial (or known) CSI, we develop an augmented-Lagrangian alternating algorithm with projected updates and provide provable sublinear convergence, with accelerated rates under mild local curvature. With unknown CSI, we perform constrained Bayesian optimization (BO) in a geometry-aware low-dimensional latent space using Gaussian process (GP) surrogates; we prove regret bounds for a constrained upper confidence bound (UCB) variant of the BO module, and demonstrate strong empirical performance of the implemented procedure. Simulations across diverse network conditions show higher detection probability at fixed false-alarm rate under adversarial attacks, large reductions in sum MSE for honest users, strong suppression of eavesdropper signal power, and fast convergence. The framework offers a practical path to secure opportunistic communication that adapts to CSI availability while coherently coordinating sensing and transmission through joint learning.
\end{abstract}


\section{Introduction}
\IEEEPARstart{T}{he} rapid proliferation of connected devices, artificial intelligence, and emerging applications such as augmented reality and autonomous systems has significantly increased the demand for faster and more efficient wireless networks~\cite{Chowdhury2020}. Realizing these capabilities requires addressing critical challenges related to spectrum efficiency, coverage, adaptability, and, increasingly, security.

Efficient and secure spectrum utilization is essential in 6G networks, particularly with the adoption of high-frequency millimeter-wave and terahertz bands. While these bands offer abundant bandwidth, they suffer from severe propagation losses and susceptibility to blockages. Traditional static spectrum allocation techniques are inadequate in the dynamic and heterogeneous 6G environment, leading to inefficient resource utilization. To address this, dynamic spectrum sharing (DSS) has emerged as a promising solution, enabling real-time access to spectrum by allowing secondary users (SUs) to opportunistically utilize licensed spectrum without interfering with primary users (PUs)~\cite{Taherpour2024}. However, this openness introduces new vulnerabilities. Malicious SUs may act as Byzantine attackers during spectrum sensing or as eavesdroppers during opportunistic transmission.

To mitigate such threats, secure distributed learning has become a cornerstone of robust DSS frameworks. By enabling collaborative and scalable decision-making, distributed learning adapts to dynamic network conditions while suppressing the influence of compromised or malicious nodes~\cite{Sedighi2013, Liu2020BigData, Xu2018, Jiang2022}. Adaptive learning methods, such as Bayesian learning and Bayesian optimization (BO), not only improve local sensing reliability but also enhance global consensus through probabilistic modeling~\cite{Nedic2017, Nie2017}, even in adversarial environments.

Meanwhile, intelligent reflecting surfaces (IRSs) have emerged as a transformative solution to counter the propagation limitations of high-frequency bands. Consisting of programmable meta-surfaces, an IRS can dynamically reconfigure wireless environments by controlling the phase and amplitude of reflected signals, thereby improving coverage and spectral efficiency without active power consumption~\cite{Wu2024, Shao2022}. IRS-assisted systems have also shown promise in enhancing secrecy~\cite{yu2023learning, Zhang2023}, improving energy efficiency~\cite{huang2020reconfigurable}, and supporting robust aerial platforms~\cite{Wei2023}.

Combining IRS with DSS yields a powerful synergy: while DSS enables adaptive and opportunistic spectrum access, IRS enhances signal quality and reliability through passive beam control. However, this integration must also address security concerns, as compromised SUs can exploit IRS-enhanced links to eavesdrop or disrupt communication. This necessitates secure, learning-based IRS-assisted frameworks that are resilient to malicious interference~\cite{Hameed2023, Saad2021}.

In this work, we investigate a secure IRS-assisted learning-based opportunistic communication framework for distributed wireless networks. Our approach jointly addresses the challenges of spectrum efficiency, environmental uncertainty, and adversarial behavior through Bayesian distributed learning and robust decision mechanisms.

\subsection{Related Works}
Recent research has explored machine learning techniques to enhance spectrum sensing in dynamic wireless environments. In~\cite{BB69}, a K-means clustering algorithm improved robustness in high-noise and mobile scenarios, while~\cite{BB75} proposed a stacked autoencoder that effectively extracted features from orthogonal frequency-division multiplexing (OFDM) signals. These works highlight the potential of unsupervised learning for robust signal detection.

Bayesian methods have been leveraged to improve spectrum efficiency and reduce overhead. In~\cite{BB70}, a Bayesian hierarchical model combined Gaussian processes and Markov chain Monte Carlo (MCMC) to capture spatial fading correlations. Meanwhile,~\cite{BB72} proposed a sparse Bayesian learning algorithm with variational inference for millimeter-wave multiple-input multiple-output (MIMO) channel estimation, significantly reducing pilot overhead.

Deep learning approaches have also addressed temporal variation and domain adaptation. A convolutional neural network in~\cite{BB76} captured spatial and temporal PU activity patterns under noise uncertainty. In~\cite{BB77}, transfer learning with unsupervised domain adaptation mitigated performance degradation under mismatched training-test conditions.

IRS-assisted communication has gained attention for enhancing performance in uncertain or constrained scenarios. In~\cite{BB73}, variational Bayesian inference achieved near-optimal IRS-assisted communication under dynamic conditions. For energy-constrained IoT systems,~\cite{BB74} proposed a scheme that models sporadic user activity and reduces pilot signaling.

Opportunistic communication using IRSs has also been studied to reduce signaling complexity. In~\cite{Yashvanth2023}, near-optimal throughput was achieved without explicit channel state information (CSI), and~\cite{Xu2023, Wang2022} introduced blind beamforming techniques that require only statistical signal power measurements for multi-IRS coordination.

Despite these advances, most prior efforts assume fully trusted users and overlook adversaries capable of corrupting sensing data or exploiting IRS-enhanced transmissions. In contrast, this work develops a secure, learning-based IRS-assisted framework for distributed wireless networks resilient to Byzantine attacks and eavesdropping.

\subsection{Motivation and Contributions}
Securing IRS-assisted opportunistic networks in the presence of Byzantine users and uncertain CSI requires sensing and transmission designs that operate coherently under adversarial, high-dimensional conditions. This work addresses that need with a single learning-based framework that couples resilient distributed sensing with secrecy-aware downlink optimization.

\begin{enumerate}[leftmargin=*,label=\textbf{\arabic*.}]
\item \textbf{Byzantine-resilient distributed sensing.} We design a logit-domain Bayesian update with trimmed aggregation and attention-weighted consensus, fused at the base station (BS) through a conservative min rule. Under standard graph robustness and degree conditions (e.g., $|\mathcal{N}_k|\!\ge\!2K_B{+}1$ for honest neighbors), the scheme preserves high detection probability with bounded adversarial presence and admits convergence/robustness guarantees.

\item \textbf{Unified secure IRS-assisted transmission.} We formulate the downlink design as sum mean-squared error (MSE) minimization for honest users under transmit-power and signal-leakage constraints, jointly optimizing the BS precoder, IRS phases, and receiver equalizers. We analyze the Karush–Kuhn–Tucker (KKT) conditions, provide explicit first-order expressions, and derive closed-form \emph{structures} in the high signal-to-noise ratio (SNR) and single-eavesdropper regimes (e.g., weighted ridge–nulling), together with informative performance bounds.

\item \textbf{Learning across CSI regimes.} With partial (or known) CSI, we propose an augmented-Lagrangian alternating algorithm with projected updates for the unit-modulus IRS and establish sublinear convergence under deterministic (full-CSI) gradients, as well as convergence-in-expectation under standard stochastic-gradient assumptions; accelerated rates hold under mild local curvature. With unknown CSI, we perform constrained BO in a geometry-aware low-dimensional latent space using Gaussian process (GP) surrogates; we provide regret bounds for a constrained upper confidence bound (UCB) BO variant and report strong empirical performance of the implemented procedure.

\item \textbf{Validated system gains.} Extensive simulations demonstrate fast convergence, strong resilience to Byzantine attacks, large reductions in sum MSE for honest users, suppressed leakage toward eavesdroppers, and improved secure wireless power transfer (WPT) efficiency, approaching full-CSI benchmarks and surpassing recent baselines.
\end{enumerate}

The remainder of this paper is organized as follows. Section~\ref{sec:system} introduces the system model, including the IRS-assisted downlink transmission and the distributed spectrum sensing setup. Section~\ref{sec:detection} describes the robust distributed learning and decision making process in the presence of Byzantine users. Section~\ref{sec:problem} formulates the secure IRS-assisted transmission as a constrained optimization problem. Section~\ref{sec:optimization} analyzes the problem and presents closed-form solutions for special cases. Section~\ref{sec:alternating} proposes a learning algorithm based on alternating optimization for scenarios with partially known CSI. Section~\ref{sec:bayesian} develops a Bayesian optimization framework for cases with unknown CSI. Section~\ref{sec:sim} reports simulation results, and Section~\ref{sec:conclusion} concludes the paper.


\section{System Model and Assumptions}
\label{sec:system}
\begin{figure}
    \centering
    \includegraphics[width=0.6\columnwidth, height=0.5\columnwidth]{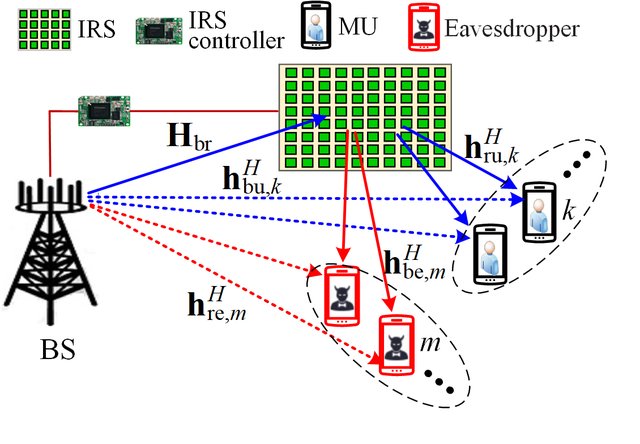}
    \caption{IRS-Assisted Cognitive Radio Network with Byzantine Users}
    \label{fig_system}
\end{figure}

We consider an IRS-assisted cognitive radio network consisting of a BS equipped with $M$ transmit antennas, an IRS with $N$ reconfigurable reflecting elements, and $K$ single-antenna SUs. The SUs include $K_H$ honest users and $K_B$ Byzantine adversaries, such that $K = K_H + K_B$. The secondary network operates under a spectrum-sharing regime with licensed PUs, who have prioritized access to the spectrum.

We define the set of all SUs as $\mathcal{K} = \{1, \ldots, K\}$, the set of Byzantine users as $\mathcal{K}_B \subset \mathcal{K}$ with $|\mathcal{K}_B| = K_B$, and the set of honest users as $\mathcal{K}_H = \mathcal{K} \setminus \mathcal{K}_B$ with $|\mathcal{K}_H| = K_H$. To ensure a strict majority of honest users, we assume $K_B \leq \left\lfloor \frac{K}{2} \right\rfloor - 1$. The set $\mathcal{N}_k \subset \mathcal{K}$ denotes the neighborhood of user $k$, i.e., the set of users that directly exchange information with user $k$ in the network graph.

Each adversarial SU $k \in \mathcal{K}_B$ is assumed to have full knowledge of the network graph $\mathcal{G} = (\mathcal{K}, \mathcal{E})$, as well as the decision-making algorithm and the current global state. These users may report arbitrary decision values $\delta_k(t) \in [0, 1]$ with the intent to mislead the fusion process.

The IRS is characterized by a diagonal phase-shift matrix $\bm{\Phi} = \operatorname{diag}(e^{\jmath\theta_1}, \ldots, e^{\jmath\theta_N})$, where $\theta_n \in [0, 2\pi)$ denotes the phase applied by the $n$-th IRS element.

The system operates in two main phases: spectrum sensing and opportunistic transmission. During the spectrum-sensing phase, all SUs collaborate to detect PU activity using energy detection over $J$ samples. If the spectrum is deemed idle, the system enters the transmission phase. The BS transmits data to honest users while Byzantine users attempt to eavesdrop. The BS employs linear precoding with a matrix $\bm{W} = [\bm{w}_1, \ldots, \bm{w}_{K_H}] \in \mathbb{C}^{M \times K_H}$, where $\bm{w}_k \in \mathbb{C}^{M \times 1}$ is the precoding vector intended for the $k$-th honest user. Each honest SU applies a complex equalization coefficient, and the overall equalizer matrix is given by $\bm{C} = \operatorname{diag}(c_1, \ldots, c_{K_H}) \in \mathbb{C}^{K_H \times K_H}$, where $c_k \in \mathbb{C}$ denotes the equalizer used by the $k$-th honest user.

The wireless channels are modeled as follows: $\bm{H} \in \mathbb{C}^{N \times M}$ denotes the channel matrix from the BS to the IRS, $\bm{h}_k \in \mathbb{C}^{N \times 1}$ denotes the channel vector from the IRS to the $k$-th honest user, and $\bm{g}_e \in \mathbb{C}^{N \times 1}$ denotes the channel vector from the IRS to the $e$-th eavesdropper, where $k = 1, \ldots, K_H$ and $e = 1, \ldots, K_B$. Fig.~\ref{fig_system} illustrates the IRS-assisted cognitive radio network under study, showing both honest and Byzantine secondary users in the system.

The received signal at the $k$-th honest user is given by
\begin{equation}
    y_k = \underbrace{\bm{h}_k^\mathrm{H} \bm{\Phi} \bm{H}}_{\triangleq \bm{h}_{\mathrm{eff},k}^\mathrm{H}} \bm{W} \bm{s} + n_k,
\end{equation}
where $\bm{s} \sim \mathcal{CN}(\bm{0}, \bm{I}_{K_H})$ is the transmitted symbol vector, and $n_k \sim \mathcal{CN}(0, \sigma_k^2)$ is the additive white Gaussian noise (AWGN) at the $k$-th user. The effective channel between the BS and the $k$-th user, incorporating the IRS, is thus defined as
\begin{equation}\label{eq:sss}
    \bm{h}_{\mathrm{eff},k} = \bm{H}^\mathrm{H} \bm{\Phi}^\mathrm{H} \bm{h}_k.
\end{equation}
For the $e$-th eavesdropper, the received signal is
\[
    y_e = \bm{g}_e^\mathrm{H} \bm{\Phi} \bm{H} \bm{W} \bm{s} + n_e,\quad n_e \sim \mathcal{CN}(0,\sigma_e^2),
\]
and the corresponding effective channel is
\[
    \bm{h}_{\mathrm{eff},e} \triangleq \bm{H}^\mathrm{H}\bm{\Phi}^\mathrm{H}\bm{g}_e \in \mathbb{C}^{M\times 1}.
\]
The effective channels of all honest users are collected in the matrix
\[
    \bm{H}_{\mathrm{eff}} \triangleq [\bm{h}_{\mathrm{eff},1},\ldots,\bm{h}_{\mathrm{eff},K_H}] \in \mathbb{C}^{M\times K_H}.
\]

Assuming linear equalization at the receiver with scalar coefficient $c_k$, the MSE at the $k$-th honest user is defined as
\begin{align} \label{eq:MSEk}
    \mathrm{MSE}_k 
    &= \mathbb{E} \left[ \left| c_k y_k - s_k \right|^2 \right]  \nonumber \\
    &= |c_k|^2 \left( \left| \bm{h}_{\mathrm{eff},k}^\mathrm{H} \bm{w}_k \right|^2 + \sum_{j\neq k}^{K_H} \left| \bm{h}_{\mathrm{eff},k}^\mathrm{H} \bm{w}_j \right|^2 + \sigma_k^2 \right) \nonumber \\
    &\quad - 2 \Re \left\{ c_k \bm{h}_{\mathrm{eff},k}^\mathrm{H} \bm{w}_k \right\} + 1,
\end{align}
where the expectation is taken over both the signal and noise distributions.

For the $e$-th eavesdropper, the received signal power is used as a proxy for information leakage and is given by
\begin{align}
    P_e^{(\mathrm{eav})} 
    &= \mathbb{E} \left[ \left| y_e \right|^2 \right] \nonumber \\
    &= \bm{g}_e^\mathrm{H} \bm{\Phi} \bm{H} \bm{W} \bm{W}^\mathrm{H} \bm{H}^\mathrm{H} \bm{\Phi}^\mathrm{H} \bm{g}_e + \sigma_e^2 \nonumber \\
    &= \sum_{k=1}^{K_H} \left| \bm{g}_e^\mathrm{H} \bm{\Phi} \bm{H} \bm{w}_k \right|^2 + \sigma_e^2,
\end{align}
We focus on controlling the \emph{signal} component of the eavesdropper's received power, defined as $P_e^{(\mathrm{sig})} \triangleq \sum_{k=1}^{K_H} \big| \bm{g}_e^\mathrm{H} \bm{\Phi} \bm{H} \bm{w}_k \big|^2$, which excludes the uncontrollable receiver noise. As will be shown in subsequent sections, both the constraint in \eqref{eq:leakage_constraint} and the penalty term in \eqref{eq:objective_function} are formulated in terms of $P_e^{(\mathrm{sig})}$.

\section{Bayesian Learning and Decision with Byzantine Resilience}
\label{sec:detection}

\begin{algorithm}[!t]
\caption{Distributed Byzantine-Resilient Belief Update}
\label{alg:ml_byzantine_attention}
\KwIn{Initial beliefs $\{\pi_k(0)\}_{k \in \mathcal{K}}$; number of rounds $\mathcal{R}$}
\KwOut{Final local belief $\delta_k(\mathcal{R})$ for each SU $k$}

\textbf{Initialization:} \\
Each SU $k$ sets logit $\psi_k(0) = \ln\!\left( \frac{\pi_k(0)}{1 - \pi_k(0)} \right)$ and initializes $\delta_k(0) = \pi_k(0)$.

\For{$t = 1, 2, \dots, \mathcal{R}$}{
    \ForEach{SU $k \in \mathcal{K}$ \textbf{(in parallel)}}{
        Observe local measurement $\boldsymbol{y}_k(t)$; \\
        Compute log-likelihood: $\ell_k(t) = \ln L_k(\boldsymbol{y}_k(t))$; \\
        Update logit: $\psi_k(t) = \psi_k(t{-}1) + \ell_k(t)$; \\
        Compute belief: $\pi_k(t) = \sig(\psi_k(t))$; \\
        
        \uIf{$|\mathcal{N}_k| \geq 2K_B + 1$}{
            Receive $\delta_i(t{-}1)$ from $i \in \mathcal{N}_k$; \\
            Sort $\{\delta_i(t{-}1)\}_{i \in \mathcal{N}_k}$ and remove the $K_B$ largest and $K_B$ smallest values; \\
            Let $\hat{\mathcal{N}}_k$ be the trimmed set; \\
            $\delta_k(t) = \min\left( \pi_k(t), \frac{1}{|\hat{\mathcal{N}}_k|} \sum_{i \in \hat{\mathcal{N}}_k} \delta_i(t{-}1) \right)$;
        }
        \Else{
            $\delta_k(t) = \pi_k(t)$;
        }
        Broadcast $\delta_k(t)$ to neighbors;
    }
}
\end{algorithm}

This section presents a Byzantine-resilient collaborative spectrum sensing framework based on Bayesian learning. Each SU maintains and updates a belief about the PU's activity while mitigating the effect of adversarial users through robust local aggregation.

\subsection{Distributed Learning and Decision}

Let $\pi_k(t) \in (0,1)$ denote the posterior belief of user $k \in \mathcal{K}$ regarding PU presence at time $t$. Upon observing a local measurement $\boldsymbol{y}_k(t)$, the belief is updated according to Bayes' rule in (\ref{eq:belief_ratio_form}):
\begin{figure*}[!b]        
  \hrulefill               
  \vspace{1ex}             
  \begin{equation}
    \pi_k(t)
    = \frac{\pi_k(t{-}1)\,g_k\bigl(\boldsymbol{y}_k(t)\mid\mathcal{H}_1\bigr)}
           {(1-\pi_k(t{-}1))\,g_k\bigl(\boldsymbol{y}_k(t)\mid\mathcal{H}_0\bigr)
            + \pi_k(t{-}1)\,g_k\bigl(\boldsymbol{y}_k(t)\mid\mathcal{H}_1\bigr)}
    = \frac{\pi_k(t{-}1)\,L_k\bigl(\boldsymbol{y}_k(t)\bigr)}
           {1-\pi_k(t{-}1) + \pi_k(t{-}1)\,L_k\bigl(\boldsymbol{y}_k(t)\bigr)},
    \label{eq:belief_ratio_form}
  \end{equation}
\end{figure*}
where $g_k(\cdot \mid \mathcal{H}_0)$ and $g_k(\cdot \mid \mathcal{H}_1)$ are the likelihood functions under hypotheses $\mathcal{H}_0$ (PU absent) and $\mathcal{H}_1$ (PU present), and the likelihood ratio is defined as:
\begin{equation}
    L_k(\boldsymbol{y}_k(t)) \triangleq \frac{g_k(\boldsymbol{y}_k(t) \mid \mathcal{H}_1)}{g_k(\boldsymbol{y}_k(t) \mid \mathcal{H}_0)}. \label{eq:likelihood_ratio}
\end{equation}

An alternative representation of the belief update uses the inverse form:
\begin{equation}
    \pi_k^{-1}(t) = 1 + \left( \pi_k^{-1}(t{-}1) - 1 \right) \cdot L_k^{-1}(\boldsymbol{y}_k(t)), \label{eq:inverse_belief}
\end{equation}
where $(\cdot)^{-1}$ denotes the reciprocal (e.g., $L_k^{-1}(y)=1/L_k(y)$).

To improve numerical stability and enable efficient ML-based implementations, we adopt the logit-domain update. The logit $\psi_k(t)$ denotes the log-odds of PU presence:
\begin{IEEEeqnarray}{rCl}
    \psi_k(t) &=& \psi_k(t{-}1) + \ell_k(t), \label{eq:logit_update}
\end{IEEEeqnarray}
where
\begin{IEEEeqnarray}{rCl}
    \psi_k(t) &\triangleq& \logit(\pi_k(t)) = \ln\!\left(\frac{\pi_k(t)}{1 - \pi_k(t)}\right), \label{eq:logit_def} \\
    \ell_k(t) &\triangleq& \ln L_k(\boldsymbol{y}_k(t)) \label{eq:log_likelihood}
\end{IEEEeqnarray}

The belief can be recovered from the logit via the logistic sigmoid function:
\begin{equation}
    \pi_k(t) = \sig(\psi_k(t)) = \frac{1}{1 + \exp(-\psi_k(t))}. \label{eq:logistic}
\end{equation}

To ensure Byzantine robustness, each SU maintains both the local belief $\pi_k(t)$ and a shared decision variable $\delta_k(t)$ used for robust consensus. The update rule for $\delta_k(t)$ depends on the neighborhood size:
\begin{equation}
  {\small
  \delta_k(t)=
  \begin{cases}
    \pi_k(t), & \text{if }|\mathcal{N}_k|<2K_B+1,\\[0.8ex]
    \displaystyle
      \min\!\Bigl(\pi_k(t),\,
        \tfrac{1}{|\hat{\mathcal{N}}_k|}
        \sum_{i\in\hat{\mathcal{N}}_k}\delta_i(t-1)\Bigr),
      & \text{otherwise.}
  \end{cases}
  }
  \label{eq:shared_update}
\end{equation}
where $\hat{\mathcal{N}}_k$ is the pruned set of neighbors obtained by removing the $K_B$ largest and $K_B$ smallest values from $\{\delta_i(t{-}1)\}_{i \in \mathcal{N}_k}$ (so $|\hat{\mathcal{N}}_k|\ge 1$ when $|\mathcal{N}_k|\ge 2K_B{+}1$). This trimmed mean suppresses extreme values to mitigate the influence of up to $K_B$ Byzantine nodes and the $\min(\cdot,\cdot)$ cap prevents adversaries from inflating the decision above the local posterior.

The distributed learning is summarized in Algorithm~\ref{alg:ml_byzantine_attention}. 

The convergence and robustness properties of the above update under Byzantine attacks are formally analyzed in Appendix~A.

\subsection{Final Collaborative Decision with Attention and Min Rule}
\label{subsec:final_decision}

To reduce overhead and enable real-time decision making, the BS performs a single measurement at the end of $\mathcal{R}$ belief sharing rounds. Here, $i_F$ denotes the frame index with $i_F = 1, 2, \ldots$, representing successive sensing frames. In each frame, sensing is carried out, and if an opportunity is detected, opportunistic transmission is performed.

Let $\ell_{BS}(i_F)$ be the log likelihood based on the BS’s local measurement. The BS maintains its logit-based belief as:
\begin{equation}
    \psi_{BS}(i_F) = \psi_{BS}(i_F{-}1) + \ell_{BS}(i_F),
    \label{eq:bs_logit_update}
\end{equation}
and similar to SUs, recovers the belief via the sigmoid function.

The BS then collects the final beliefs $\{\delta_k(\mathcal{R})\}_{k \in \mathcal{K}}$, trims the $K_B$ smallest and $K_B$ largest values, and defines the set of remaining trusted users as $\tilde{\mathcal{K}}_H \subseteq \mathcal{K}$.

The average of the trusted beliefs is:
\begin{equation}
    \bar{\delta} = \frac{1}{|\tilde{\mathcal{K}}_H|} \sum_{j \in \tilde{\mathcal{K}}_H} \delta_j(\mathcal{R}).
\end{equation}

To emphasize users close to the consensus, attention weights are computed using a softmax-style function:
\begin{equation}
    \alpha_k = \frac{\exp\left(-\frac{(\delta_k(\mathcal{R}) - \bar{\delta})^2}{\tau_a}\right)}{\sum_{j \in \tilde{\mathcal{K}}_H} \exp\left(-\frac{(\delta_j(\mathcal{R}) - \bar{\delta})^2}{\tau_a}\right)}, \quad \forall k \in \tilde{\mathcal{K}}_H,
\end{equation}
where $\tau_a > 0$ is the attention temperature that controls the sharpness of the weighting.

The BS computes the attention-weighted consensus:
\begin{equation}
    \delta_{BS}^{\mathrm{att}}(i_F) = \sum_{k \in \tilde{\mathcal{K}}_H} \alpha_k \, \delta_k(\mathcal{R}).
\end{equation}

To ensure robustness, the BS fuses its own belief with the collaborative input using a conservative min-rule:
\begin{equation}
    \delta_{BS}(i_F) = \min \left\{ \delta_{BS}^{\mathrm{att}}(i_F), \; \pi_{BS}(i_F) \right\}.
    \label{eq:bs_fusion_attention_min}
\end{equation}

The final decision is made by thresholding:
\begin{equation}
    \mathcal{H}^* = 
    \begin{cases}
        \mathcal{H}_1, & \text{if } \delta_{BS}(i_F) \geq \tau, \\
        \mathcal{H}_0, & \text{otherwise},
    \end{cases}
    \label{eq:final_decision_updated}
\end{equation}
where $\tau \in (0, 1)$ is the decision threshold.

\section{Opportunistic Secure Transmission Optimization}
\label{sec:problem}
In Section~\ref{sec:detection}, we designed a Byzantine-resilient, logit-based spectrum-sensing protocol that fuses trimmed local beliefs into a reliable decision on PU activity.  In this section, we leverage that detection outcome to formulate and solve a secure transmission problem: jointly optimizing the BS precoder, IRS phase-shift matrix, and SU equalizers to minimize sum-MSE at honest users while enforcing both a total power budget and a hard cap on signal leakage toward Byzantine eavesdroppers.

We aim to jointly optimize the BS precoding matrix $\boldsymbol{W} \in \mathbb{C}^{M \times K_H}$, the IRS phase shift matrix $\boldsymbol{\Phi} = \operatorname{diag}(e^{\J \theta_1}, \ldots, e^{\J \theta_N})$, and the receiver equalizer coefficients $\boldsymbol{C} = \operatorname{diag}(c_1, \ldots, c_{K_H})$, to minimize a composite objective function that incorporates communication reliability, secrecy enhancement, and resource constraints. The secure transmission optimization problem is formulated as
\begin{alignat}{2}
  \min_{\boldsymbol{W},\boldsymbol{\Phi},\boldsymbol{C}}\ & \mathcal{F}(\boldsymbol{W},\boldsymbol{\Phi},\boldsymbol{C}) \label{eq:secure_opt_obj} \\
  \text{s.t.}\ & \sum_{e=1}^{K_B}\sum_{k=1}^{K_H} \big| \boldsymbol{g}_e^{\!H} \boldsymbol{\Phi} \boldsymbol{H} \boldsymbol{w}_k \big|^2 
    \le \Gamma_{\mathrm{leak}}, \label{eq:leakage_constraint} \\
  & \sum_{k=1}^{K_H} \|\boldsymbol{w}_k\|^2 \le P_{\mathrm{max}}, \label{eq:power_constraint} \\
  & |[\boldsymbol{\Phi}]_{n,n}|=1,\ \forall n=1,\ldots,N, \label{eq:unit_modulus_constraint} \\
  & c_k\in\mathbb{C},\ \forall k=1,\ldots,K_H. \label{eq:ck_constraint}
\end{alignat}

Here, $P^{(\mathrm{sig})}_e \triangleq \sum_{k=1}^{K_H} \big| \boldsymbol{g}_e^\mathrm{H} \boldsymbol{\Phi} \boldsymbol{H} \boldsymbol{w}_k \big|^2$ is the \emph{signal-leakage} power toward eavesdropper $e$. Note that the total received power at the eavesdropper is $P_e^{(\mathrm{eav})} = P_e^{(\mathrm{sig})} + \sigma_e^2$; the controllable part is $P^{(\mathrm{sig})}_e$, which appears in both the constraint and the (optional) penalty below.

The objective function $\mathcal{F}(\boldsymbol{W}, \boldsymbol{\Phi}, \boldsymbol{C})$ quantifies the trade-off between communication reliability, eavesdropping protection, and power consumption, and is expressed as
\begin{equation}
\begin{split}
\mathcal{F}(\boldsymbol{W},\boldsymbol{\Phi},\boldsymbol{C})
&= \sum_{k=1}^{K_H} \mathrm{MSE}_k(\boldsymbol{W},\boldsymbol{\Phi},c_k)
  + \lambda \sum_{e=1}^{K_B} P^{(\mathrm{sig})}_e\\
&\quad + \mu \|\boldsymbol{W}\|_F^2,
\end{split}
\label{eq:objective_function}
\end{equation}

where $\lambda$ and $\mu$ are positive weighting factors that balance secrecy pressure, reliability, and power regularization. The leakage penalty uses the \emph{non-decreasing} function $P^{(\mathrm{sig})}_e$, which correctly discourages leakage; the hard cap \eqref{eq:leakage_constraint} further prevents approaching high-leakage solutions.

\textbf{Remark:} In \eqref{eq:objective_function}, when $\boldsymbol{C} = \boldsymbol{I}$ (i.e., no receiver equalization), the MSE terms simplify to received-power terms, and the formulation reduces to a secure WPT/SWIPT design where the secrecy-aware power shaping is governed by the joint presence of the \emph{soft} leakage penalty $\lambda \sum_e P^{(\mathrm{sig})}_e$ and the \emph{hard} leakage constraint \eqref{eq:leakage_constraint}.

\section{Optimization Problem Analysis and Closed-Form Solutions}
\label{sec:optimization}
We analyze the structure of the general optimization problem and its properties, and derive closed-form solutions for special cases with analytical tractability. These serve as benchmarks and offer intuition on the behavior of the proposed learning algorithms.

We focus on the secure transmission problem in \eqref{eq:secure_opt_obj}, which minimizes the total MSE at honest users subject to secrecy and power constraints, with variables \(\mathcal{X}=(\boldsymbol{W},\boldsymbol{\Phi},\boldsymbol{C})\).

To handle the non-convex constraints, we formulate the Lagrangian as
\begin{equation}
\begin{aligned}
\mathcal{L}(\mathcal{X}, \bm{\lambda})
&= \sum_{k=1}^{K_H} \mathrm{MSE}_k(\mathcal{X})
   + \mu \|\bm{W}\|_F^2
   + \lambda \sum_{e=1}^{K_B} P_e^{(\mathrm{sig})}(\mathcal{X}) \\[0.5ex]
&\quad + \lambda_1\Bigl(
       \sum_{e=1}^{K_B} P_e^{(\mathrm{sig})}(\mathcal{X})
       - \Gamma_{\mathrm{leak}}
   \Bigr) \\[0.5ex]
&\quad + \lambda_2\Bigl(
       \|\bm{W}\|_F^2
       - P_{\mathrm{max}}
   \Bigr)\,,
\end{aligned}
\label{eq:lagrangian}
\end{equation}
where $\lambda,\mu>0$ are weighting factors in the objective and $\lambda_1, \lambda_2 \geq 0$ are dual variables associated with the secrecy-leakage and transmit-power constraints, respectively. Here,
\[
P_e^{(\mathrm{sig})}(\mathcal{X}) \triangleq \sum_{k=1}^{K_H}
\big|\bm{g}_e^{\mathrm{H}} \bm{\Phi} \bm{H} \bm{w}_k\big|^2
\]
denotes the \emph{signal} component of the eavesdropper's received power (the controllable part).

The necessary optimality conditions are given by the KKT system. The stationarity condition with respect to $\bm{W}$ yields
\begin{equation}
\sum_{k=1}^{K_H} \nabla_{\bm{W}} \mathrm{MSE}_k
\;+\; 2(\mu+\lambda_2)\,\bm{W}
\;+\; (\lambda+\lambda_1) \sum_{e=1}^{K_B} \nabla_{\bm{W}} P_e^{(\mathrm{sig})}
= \bm{0}.
\label{eq:kkt_w}
\end{equation}

For the IRS phase variables $\theta_n$, the stationarity condition is obtained by differentiating the Lagrangian with respect to each $\theta_n$, which gives
\begin{equation}
\frac{\partial}{\partial \theta_n} \!\left[ \sum_{k=1}^{K_H} \mathrm{MSE}_k(\mathcal{X}) \right]
\;+\; (\lambda+\lambda_1) \sum_{e=1}^{K_B}
\frac{\partial P_e^{(\mathrm{sig})}(\mathcal{X})}{\partial \theta_n} = 0,
\quad \forall n.
\label{eq:kkt_phi}
\end{equation}

For the equalizer coefficients $c_k$, using Wirtinger calculus with respect to $c_k^{*}
$ gives the first-order condition
\begin{equation}
\begin{split}
\frac{\partial\,\mathrm{MSE}_k}{\partial c_k^*}
&= c_k \Bigl(\|\bm{h}_{\mathrm{eff},k}^{H}\bm{W}\|_2^2 + \sigma_k^2\Bigr)
  - \bm{h}_{\mathrm{eff},k}^{H}\bm{w}_k \\[0.5ex]
&= 0, \quad \forall\,k,
\end{split}
\label{eq:kkt_ck}
\end{equation}
where $\bm{h}_{\mathrm{eff},k} \triangleq \bm{H}^\mathrm{H}\bm{\Phi}^\mathrm{H}\bm{h}_k$ denotes the effective channel between the BS and user $k$, including the IRS path.

\noindent\textit{Explicit form of $\sum_k\nabla_{\bm{W}}\mathrm{MSE}_k$:}
keeping $(\bm{\Phi},\bm{C})$ fixed and letting
$\bm{H}_{\mathrm{eff}}\!\triangleq\![\bm{h}_{\mathrm{eff},1},\ldots,\bm{h}_{\mathrm{eff},K_H}]$,
one obtains
\begin{equation}
\sum_{k=1}^{K_H}\!\nabla_{\bm{W}}\mathrm{MSE}_k
\;=\; 2\,\bm{H}_{\mathrm{eff}}\!\,\mathrm{diag}(|\bm{c}|^2)\,\bm{H}_{\mathrm{eff}}^{\!H}\bm{W}
\;-\; 2\,\bm{H}_{\mathrm{eff}}\bm{C}^{H},
\label{eq:gradW_MSE}
\end{equation}
where $\bm{c}=[c_1,\ldots,c_{K_H}]^\top$.

The complementary slackness conditions are written as
\begin{align}
\lambda_1 \left( \sum_{e=1}^{K_B} P_e^{(\mathrm{sig})}(\mathcal{X}) - \Gamma_{\mathrm{leak}} \right) &= 0,
\label{eq:kkt_leak} \\
\lambda_2 \left( \|\bm{W}\|_F^2 - P_{\mathrm{max}} \right) &= 0.
\label{eq:kkt_power}
\end{align}

The conditions \eqref{eq:kkt_w}--\eqref{eq:kkt_power} together characterize any stationary point of the constrained problem under mild regularity assumptions. The optimization problem is inherently non-convex due to the unit-modulus structure of $\bm{\Phi}$ and the bilinear coupling between $\bm{W}$ and $\bm{\Phi}$ in both the MSE and secrecy-leakage terms. Moreover, the signal-leakage functions $P_e^{(\mathrm{sig})}(\mathcal{X})$ are quadratic in $\bm{W}$ and nonlinear in $\bm{\Phi}$, further complicating the structure.

Despite the non-convexity, the structure of the problem allows for decomposition into more tractable subproblems. For example, fixing $(\bm{\Phi}, \bm{C})$ reduces the optimization over $\bm{W}$ to a regularized quadratic program, while fixing $(\bm{W}, \bm{C})$ yields a unit-modulus constrained problem in $\bm{\Phi}$, which can be handled via manifold optimization or successive convex approximation. Similarly, fixing $(\bm{W}, \bm{\Phi})$ leads to a decoupled update for each $c_k$, with the closed-form solution obtained from \eqref{eq:kkt_ck}:
\begin{equation}
c_k^{\star} = \frac{ \bm{h}_{\mathrm{eff},k}^\mathrm{H} \bm{w}_k }{ \sum_{j=1}^{K_H} \left| \bm{h}_{\mathrm{eff},k}^\mathrm{H} \bm{w}_j \right|^2 + \sigma_k^2 }.
\label{eq:ck_closed_form}
\end{equation}

This alternating structure enables efficient block-wise optimization, where the dual variables $\lambda_1, \lambda_2$ are updated via subgradient methods. In practice, convergence to a stationary point is achieved under mild regularity assumptions, and the resulting solutions exhibit favorable performance with respect to both MSE minimization and secrecy-leakage suppression.

\subsection{Closed-Form Solutions for Special Cases}
\label{sec:closed_form_special_cases}

We now investigate analytically tractable regimes where closed-form solutions can be derived. These cases shed light on the structure of the optimal solutions and enable practical system design insights.

\subsubsection{Case 1: High-SNR Regime ($\sigma_k^2 \to 0$)}

In the high-SNR regime, the thermal noise becomes negligible. With the MMSE equalizer $c_k=c_k^\star$ in \eqref{eq:ck_closed_form} and $\sigma_k^2 \to 0$, the per-user MSE reduces to one minus the fraction of desired power to total received power. The joint design therefore reduces to
\begin{align}
\min_{\bm{W}, \bm{\Phi}} \quad 
& \sum_{k=1}^{K_H} \left( 1 - 
\frac{\left|\bm{h}_{\mathrm{eff},k}^\mathrm{H} \bm{w}_k\right|^2}
{\sum_{i=1}^{K_H} \left|\bm{h}_{\mathrm{eff},k}^\mathrm{H} \bm{w}_i\right|^2} \right) \label{eq:high_snr_obj} \\
\text{s.t.} \quad 
& \sum_{e=1}^{K_B} \big\| \bm{g}_e^\mathrm{H} \bm{\Phi} \bm{H} \bm{W} \big\|_2^2 \leq \Gamma_{\mathrm{leak}}, \nonumber \\
& \|\bm{W}\|_F^2 \leq P_{\mathrm{max}}. \nonumber
\end{align}

\begin{lemma}
\label{lem:highsnr}
Under high-SNR conditions and fixed IRS phase matrix $\bm{\Phi}$, a precoder aligned with the dominant \emph{left} singular subspace of the effective concatenated channel is optimal within the interference-free limit. Let the singular value decomposition (SVD) be
\(
\bm{H}_{\mathrm{eff}} = \bm{U}\bm{\Sigma}\bm{V}^\mathrm{H}
\)
with
\(
\bm{H}_{\mathrm{eff}} \triangleq \bm{H}^\mathrm{H} \bm{\Phi}^\mathrm{H} [\bm{h}_1, \dots, \bm{h}_{K_H}]\in\mathbb{C}^{M\times K_H}.
\)
Then an optimal structure is
\begin{equation}
\bm{W}^* = \bm{U}_{:,1:K_H}\,\bm{\Sigma}',
\end{equation}
where $\bm{U}_{:,1:K_H}$ collects the first $K_H$ left singular vectors and $\bm{\Sigma}'$ is a diagonal power-allocation matrix selected to satisfy the transmit-power and secrecy constraints.
\end{lemma}

\begin{proof}
At high SNR, maximizing each user's desired-to-total power ratio reduces to placing the beams in the dominant transmit subspace of $\bm{H}_{\mathrm{eff}}$, i.e., the span of its left singular vectors. The SVD $\bm{H}_{\mathrm{eff}}=\bm{U}\bm{\Sigma}\bm{V}^\mathrm{H}$ diagonalizes the corresponding Gram operator; choosing $\bm{W}$ in $\mathrm{span}(\bm{U}_{:,1:K_H})$ preserves optimality under power reallocation. The diagonal $\bm{\Sigma}'$ implements power allocation subject to the constraints.
\end{proof}

\subsubsection{Case 2: Single Eavesdropper and Multiple Honest Users ($K_B=1$, $K_H$ arbitrary)}
\label{subsec:single_eav_multi_user}

Consider the scenario with a single eavesdropper ($K_B=1$) and multiple honest users ($K_H \geq 1$). Let $\bm{g} \triangleq \bm{g}_1$ denote the IRS-to-eavesdropper channel and define the effective BS-to-eavesdropper channel as $\bm{h}_{\mathrm{eff},e} \triangleq \bm{H}^\mathrm{H}\bm{\Phi}^\mathrm{H}\bm{g} \in \mathbb{C}^{M\times 1}$. The secrecy constraint simplifies to
\begin{equation}
\big\|\bm{g}^\mathrm{H} \bm{\Phi} \bm{H} \bm{W}\big\|_2^2 \leq \Gamma_{\mathrm{leak}},
\label{eq:single_eav_constraint}
\end{equation}
i.e., the \emph{signal-leakage} power toward the eavesdropper is capped.

\begin{lemma}
\label{thm:single_eav_multi_user}
For fixed IRS phase matrix $\bm{\Phi}$, the optimal precoder structure for $K_B=1$ and arbitrary $K_H$ admits the ridge–nulling form
\begin{equation}
\label{eq:w_star_single_eav_correct}
\begin{split}
\bm{W}^* =
&\Bigl(
  \bm{H}_{\mathrm{eff}}\mathrm{diag}(|\bm{c}|^2)\bm{H}_{\mathrm{eff}}^{H}
  + \gamma\,\bm{I}_M\\
&\quad + (\lambda+\lambda_1)\,\bm{h}_{\mathrm{eff},e}\bm{h}_{\mathrm{eff},e}^{H}
\Bigr)^{-1}
\bm{H}_{\mathrm{eff}}\,\bm{C}^{H}.
\end{split}
\end{equation}
where $\bm{H}_{\mathrm{eff}} = [\bm{h}_{\mathrm{eff},1},\ldots,\bm{h}_{\mathrm{eff},K_H}] \in \mathbb{C}^{M\times K_H}$ with $\bm{h}_{\mathrm{eff},k} \triangleq \bm{H}^\mathrm{H}\bm{\Phi}^\mathrm{H}\bm{h}_k$, the vector $\bm{c}=[c_1,\ldots,c_{K_H}]^\top$ collects the equalizers ($\bm{C}=\mathrm{diag}(c_1,\ldots,c_{K_H})$), $\gamma \triangleq \mu+\lambda_2$, and $\lambda_1 \ge 0$ is chosen to satisfy \eqref{eq:single_eav_constraint} with equality when active.
\end{lemma}

\begin{proof}
Starting from the Lagrangian in \eqref{eq:lagrangian} with $K_B=1$ and using the explicit gradient identity in \eqref{eq:gradW_MSE}, the stationarity condition with respect to $\bm{W}$ yields
\[
\begin{split}
\Big(
  \bm{H}_{\mathrm{eff}}\mathrm{diag}(|\bm{c}|^2)\bm{H}_{\mathrm{eff}}^\mathrm{H}
  &+ (\mu+\lambda_2)\bm{I}_M\\
  &+ (\lambda+\lambda_1)\,\bm{h}_{\mathrm{eff},e}\bm{h}_{\mathrm{eff},e}^\mathrm{H}
\Big)\bm{W}
= \bm{H}_{\mathrm{eff}}\,\bm{C}^\mathrm{H}.
\end{split}
\]

The eavesdropper term arises from differentiating both the soft leakage penalty ($\lambda$) and the hard leakage constraint via its multiplier ($\lambda_1$). Solving for $\bm{W}$ gives \eqref{eq:w_star_single_eav_correct}. The multiplier $\lambda_1$ adjusts to enforce \eqref{eq:single_eav_constraint}.
\end{proof}

\begin{corollary}
When the secrecy constraint is inactive ($\lambda_1=0$), \eqref{eq:w_star_single_eav_correct} becomes
\[
\bm{W}^* =
\Big(\bm{H}_{\mathrm{eff}}\mathrm{diag}(|\bm{c}|^2)\bm{H}_{\mathrm{eff}}^\mathrm{H}
  + \gamma\,\bm{I}_M
  + \lambda\,\bm{h}_{\mathrm{eff},e}\bm{h}_{\mathrm{eff},e}^{H}\Big)^{-1}
\bm{H}_{\mathrm{eff}}\bm{C}^\mathrm{H},
\]
so the \emph{soft} penalty $\lambda$ still pushes the solution away from the eavesdropper’s direction. If, in addition, $\lambda=0$, it reduces to a regularized MMSE/ZF form
\(
\big(\bm{H}_{\mathrm{eff}}\mathrm{diag}(|\bm{c}|^2)\bm{H}_{\mathrm{eff}}^\mathrm{H} + \gamma\,\bm{I}_M\big)^{-1}
\bm{H}_{\mathrm{eff}}\bm{C}^\mathrm{H}.
\)
When the secrecy constraint dominates ($\lambda_1 \to \infty$), the solution approaches strict null-steering with respect to $\bm{h}_{\mathrm{eff},e}$ while still fitting $\bm{H}_{\mathrm{eff}}\bm{C}^{H}$ under the power budget.
\end{corollary}

\subsection{Performance Bounds}
\label{subsec:performance_bounds}

To evaluate the quality and scalability of our solution, we now derive analytical bounds for the optimal objective value and asymptotic MSE performance.

\begin{theorem}[Lower Bound on Optimal Cost]
\label{thm:lower_bound}
Let $p^*$ denote the optimal value of problem~\eqref{eq:secure_opt_obj}. Then
\begin{equation}
p^* \;\ge\; 
\min_{\{p_k\ge 0\}}
\ \sum_{k=1}^{K_H}
\frac{\sigma_k^2}{\|\bm{h}_{\mathrm{eff},k}\|_2^2\, p_k + \sigma_k^2}
\quad\text{s.t.}\quad
\sum_{k=1}^{K_H} p_k \le P_{\mathrm{max}}.
\end{equation}
\end{theorem}

\begin{proof}
Consider the relaxation that allocates nonnegative powers $\{p_k\}$ to user directions $\bm{h}_{\mathrm{eff},k}$ independently, with a shared budget $\sum_k p_k\le P_{\mathrm{max}}$. With linear MMSE receivers, a per-user lower bound is
\(
\mathrm{MSE}_k \ge \frac{\sigma_k^2}{\|\bm{h}_{\mathrm{eff},k}\|_2^2 p_k + \sigma_k^2}.
\)
Summing these bounds and minimizing over all feasible $\{p_k\}$ yields a valid lower bound on the \emph{sum} MSE, since any feasible $(\bm{W},\bm{\Phi},\bm{C})$ induces some power allocation that satisfies the budget. The leakage penalty is nonnegative and the leakage cap does not reduce this universal lower bound.
\end{proof}

\begin{theorem}[High-SNR Asymptotics]
\label{thm:high_snr}
When $P_{\mathrm{max}} \to \infty$, the optimal MSE for user $k$ behaves as
\begin{equation}
\mathrm{MSE}_k^* = \frac{1}{\mathrm{SINR}_k^*} + \mathcal{O} \left( \frac{1}{(\mathrm{SINR}_k^*)^2} \right),
\end{equation}
where $\mathrm{SINR}_k^*$ is the optimal post-equalization signal-to-interference-plus-noise ratio (SINR),
\[
\mathrm{SINR}_k
= \frac{\big|\bm{h}_{\mathrm{eff},k}^\mathrm{H}\bm{w}_k\big|^2}
       {\sum_{j\neq k}\big|\bm{h}_{\mathrm{eff},k}^\mathrm{H}\bm{w}_j\big|^2 + \sigma_k^2},
\qquad
\mathrm{SINR}_k^* \triangleq \max_{\bm{W},\bm{\Phi}} \mathrm{SINR}_k.
\]
\end{theorem}

\begin{proof}
With $c_k^\star$ in \eqref{eq:ck_closed_form}, the per-user MSE is
\[
\mathrm{MSE}_k = 1 - \frac{\left|\bm{h}_{\mathrm{eff},k}^\mathrm{H}\bm{w}_k\right|^2}{\sum_{j=1}^{K_H}\left|\bm{h}_{\mathrm{eff},k}^\mathrm{H}\bm{w}_j\right|^2 + \sigma_k^2}
= \frac{1}{1+\mathrm{SINR}_k}.
\]
As $P_{\mathrm{max}}\to\infty$, the effective signal/interference terms dominate the noise, and the Taylor expansion of $1/(1+\mathrm{SINR}_k^*)$ gives the stated asymptotic series.
\end{proof}

\section{Partial CSI: Alternating Optimization Learning}
\label{sec:alternating}

The non-convex problem in \eqref{eq:secure_opt_obj} arises in the context of secure IRS-aided communication design under the full-CSI viewpoint adopted in Section~\ref{sec:problem}. In many realistic scenarios, however, perfect or instantaneous CSI is unavailable due to estimation errors, feedback delays, or hardware constraints. To address these challenges, we propose a self-supervised, gradient-based alternating optimization approach with adaptive step sizes that \emph{does not require explicit full CSI}. Instead, it leverages partially known CSI or observed data (e.g., pilot-based measurements, received-symbol statistics, or leakage proxies) to iteratively improve the system parameters. Consistent with Section~\ref{sec:problem}, the IRS unit-modulus constraint $|[\bm{\Phi}]_{n,n}|=1$ is treated as a \emph{hard} constraint enforced via projection (we do not include a soft penalty on $|[\bm{\Phi}]_{n,n}|-1$ in the objective).

\subsection{Proposed Alternating Learning Algorithm}

At iteration $t$, the transmit precoder $\bm{W}$ and the receive equalizer $\bm{C}$ are updated while keeping the IRS phase matrix $\bm{\Phi}^{(t)}$ fixed; subsequently, the IRS phases are updated via a projected step to satisfy the unit-modulus constraints. This decomposes the design into the block subproblems
\begin{align}
    (\bm{W}^{(t+1)}, \bm{C}^{(t+1)}) &= \arg\min_{\bm{W}, \bm{C}} \ \mathcal{F}(\bm{W}, \bm{\Phi}^{(t)}, \bm{C}), \label{eq:updateW_C}\\
    \bm{\Phi}^{(t+1)} &= \arg\min_{\bm{\Phi}} \ \mathcal{F}(\bm{W}^{(t+1)}, \bm{\Phi}, \bm{C}^{(t+1)}), \label{eq:updatePhi}
\end{align}
where $\mathcal{F}$ is the objective in \eqref{eq:objective_function} (with \eqref{eq:leakage_constraint}–\eqref{eq:power_constraint} enforced via dual/penalty terms and projection for the unit-modulus constraint).

To incorporate constraints with first-order updates, we employ an augmented Lagrangian
\begin{equation}
\begin{aligned}
\mathcal{L}_{\mathrm{aug}}(\bm{W}, \bm{\Phi}, \bm{C}, \bm{\lambda})
&= \mathcal{F}(\bm{W}, \bm{\Phi}, \bm{C}) \\[0.5ex]
&\quad + \sum_{i=1}^{2}\Bigl(\lambda_i\,g_i(\bm{W},\bm{\Phi})
   + \tfrac{\rho}{2}\,g_i(\bm{W},\bm{\Phi})^2\Bigr)\,.
\end{aligned}
\label{eq:augmented_lagrangian}
\end{equation}
The soft secrecy pressure $+\lambda\sum_e P_e^{(\mathrm{sig})}$ in $\mathcal{F}$ acts in tandem with the hard cap $g_1(\bm{W},\bm{\Phi})\le 0$, which is also penalized in the augmented Lagrangian.

with constraint functions corresponding to \eqref{eq:leakage_constraint}–\eqref{eq:power_constraint}:
\begin{align}
g_1(\bm{W},\bm{\Phi}) &\triangleq \sum_{e=1}^{K_B}\sum_{k=1}^{K_H} \left| \bm{g}_e^\mathrm{H}\bm{\Phi}\bm{H}\bm{w}_k \right|^2 - \Gamma_{\mathrm{leak}}, \\
g_2(\bm{W},\bm{\Phi}) &\triangleq \|\bm{W}\|_F^2 - P_{\mathrm{max}},
\end{align}
dual variables $\bm{\lambda}=[\lambda_1,\lambda_2]^\top \succeq \bm{0}$, and penalty parameter $\rho>0$. In the partial-CSI setting, the gradients (or subgradients) of $\mathcal{L}_{\mathrm{aug}}$ are \emph{estimated} from observed/estimated feedback (e.g., sample MSEs at honest users and leakage power proxies measured at cooperative monitors or inferred from statistics). Under full CSI, the same template uses \emph{exact} gradients.

We then carry out block-wise (projected) gradient steps with a diminishing step-size schedule:
\begin{align}
(\bm{W}^{(t+1)}, \bm{C}^{(t+1)}) &= \arg\min_{\bm{W}, \bm{C}} \ \mathcal{L}_{\mathrm{aug}}(\bm{W}, \bm{\Phi}^{(t)}, \bm{C}, \bm{\lambda}^{(t)}), \\
\bm{\Phi}^{(t+1)} &= \arg\min_{\bm{\Phi}} \ \mathcal{L}_{\mathrm{aug}}(\bm{W}^{(t+1)}, \bm{\Phi}, \bm{C}^{(t+1)}, \bm{\lambda}^{(t)}),
\end{align}
where each ``$\arg\min$'' is realized by one or a few (stochastic) gradient steps. The projection for $\bm{\Phi}$ \emph{zeros off-diagonal entries and normalizes the diagonal to unit modulus}:
\begin{align}
\mathcal{P}_{\mathcal{U}}(\bm{\Phi})
&\triangleq 
\operatorname{diag}\!\bigl(e^{\J\,\arg([\bm{\Phi}]_{1,1})},\ldots,
e^{\J\,\arg([\bm{\Phi}]_{N,N})}\bigr),
\label{eq:projU}\\
\mathcal{U}
&=
\bigl\{\bm{\Phi}:\;|[\bm{\Phi}]_{n,n}|=1,\;\forall n\bigr\}.
\label{eq:setU}
\end{align}
\textit{Remark.} Equivalently, one may parametrize $\bm{\Phi}$ directly by its phase vector $\boldsymbol{\theta}\!\in\!\mathbb{R}^N$ with $[\bm{\Phi}]_{n,n}=e^{\J \theta_n}$ and perform gradient steps on $\boldsymbol{\theta}$, which avoids creating off-diagonal artifacts.

We sometimes abbreviate $\|\mathcal{X}\|_F^2 \triangleq \|\bm{W}\|_F^2 + \|\bm{\Phi}\|_F^2 + \|\bm{C}\|_F^2$ for compactness; note that $\|\bm{\Phi}\|_F^2$ is \emph{constant} under \eqref{eq:setU} and is not used for step-size or stopping rules.

The detailed implementation is summarized below, where the step sizes follow a harmonic decay $\gamma_t=\gamma_0/(1+\alpha t)$ (satisfying Robbins--Monro conditions), and the dual variables are updated by projected subgradient ascent.

\begin{algorithm}[t]
\caption{Alternating Optimization for IRS-Assisted Secure Communication}
\label{alg:ml_enhanced_ssl}
\SetAlFnt{\small}
\SetAlCapFnt{\small}
\SetAlCapNameFnt{\small}
\DontPrintSemicolon
\LinesNumbered
\SetAlgoNlRelativeSize{-1}
\setlength{\abovecaptionskip}{4pt}
\setlength{\belowcaptionskip}{2pt}

\KwIn{Initial $\mathcal{X}^{(0)}=\{\bm{W}^{(0)},\bm{\Phi}^{(0)},\bm{C}^{(0)}\}$; duals $\bm{\lambda}^{(0)}=[\lambda_1^{(0)},\lambda_2^{(0)}]^\top\!\in\mathbb{R}^2_+$; learning params $\gamma_0,\alpha>0$; augmented Lagrangian coeff.\ $\rho>0$; tolerances $\epsilon>0$, $T_{\max}$}
\KwOut{Primal-dual $(\mathcal{X}^*,\bm{\lambda}^*)$}

$t\!\gets\!0$, $\mathcal{L}_{\mathrm{aug}}^{(0)}\!\gets\!\mathcal{L}_{\mathrm{aug}}(\mathcal{X}^{(0)},\bm{\lambda}^{(0)})$\\
\Repeat{$t\ge1$ \textbf{and} $\|\mathcal{X}^{(t)}-\mathcal{X}^{(t-1)}\|_F\le\epsilon$ \textbf{or} $t\ge T_{\max}$}{
  $\gamma_t\!\gets\!\gamma_0/(1+\alpha t)$\tcp*{Robbins--Monro}

  \ForEach{$(\bm{v},\mathcal{P})\!\in\!\{(\bm{W},\emptyset),(\bm{C},\emptyset),(\bm{\Phi},\mathcal{P}_{\mathcal{U}})\}$}{
    $\bm{v}^{(t+1)}\!\gets\!\bm{v}^{(t)} - \gamma_t\,\nabla_{\bm{v}}\mathcal{L}_{\mathrm{aug}}(\mathcal{X}^{(t)},\bm{\lambda}^{(t)})$\tcp*{stochastic estimate under partial CSI}
    \If{$\mathcal{P}\neq\emptyset$}{ $\bm{v}^{(t+1)}\!\gets\!\mathcal{P}(\bm{v}^{(t+1)})$\tcp*{project to unit-modulus for $\bm{\Phi}$} }
  }

  \[
  \bm{g}^{(t+1)} =
  \begin{bmatrix}
    \displaystyle \sum_{e=1}^{K_B}\sum_{k=1}^{K_H}\big|\bm{g}_e^{H}\bm{\Phi}^{(t+1)}\bm{H}\bm{w}_k^{(t+1)}\big|^2 - \Gamma_{\text{leak}}\\
    \|\bm{W}^{(t+1)}\|_F^2 - P_{\text{max}}
  \end{bmatrix}
  \]

  $\bm{\lambda}^{(t+1)}\!\gets\!\max\!\left(\bm{0},\,\bm{\lambda}^{(t)}+\rho\,\bm{g}^{(t+1)}\right)$\tcp*{projected ascent}

  $t\!\gets\!t+1$
}

\Return{$\mathcal{X}^{(t)},\bm{\lambda}^{(t)}$}
\end{algorithm}

\subsection{Convergence Analysis}
\label{sec:convergence}

We now establish the convergence properties of the alternating learning algorithm in Algorithm~\ref{alg:ml_enhanced_ssl}. Let $\mathcal{X}^{(t)}=(\bm{W}^{(t)},\bm{\Phi}^{(t)},\bm{C}^{(t)})$ denote the iterates at iteration $t$, and let $\mathcal{L}_{\mathrm{aug}}(\mathcal{X},\bm{\lambda})$ be the augmented Lagrangian defined in Section~\ref{sec:alternating}. Throughout, the IRS unit-modulus constraint is enforced by projection $\mathcal{P}_{\mathcal{U}}$, and the dual variables $\bm{\lambda}$ are updated via projected ascent. The step sizes follow a diminishing schedule $\{\gamma_t\}$ satisfying Robbins--Monro conditions (i.e., $\sum_t \gamma_t=\infty$ and $\sum_t \gamma_t^2<\infty$). We distinguish two regimes:
(i) \emph{Deterministic/full-CSI}: exact block gradients are used; and
(ii) \emph{Stochastic/partial-CSI}: block gradients are replaced by unbiased estimators with bounded second moments.

For constrained blocks, we measure primal stationarity via the \emph{projected gradient mapping}
\[
\mathcal{G}_{\eta}(\mathcal{X},\bm{\lambda})
\;\triangleq\; \tfrac{1}{\eta}\Big(\mathcal{X}
-\Pi_{\mathcal{F}}\!\big(\mathcal{X}-\eta\,\nabla_{\mathcal{X}}\mathcal{L}_{\mathrm{aug}}(\mathcal{X},\bm{\lambda})\big)\Big),
\]
where $\Pi_{\mathcal{F}}$ denotes the Cartesian projection onto the product set collecting the per-block constraints (including $\mathcal{U}$ for $\bm{\Phi}$).

\begin{theorem}[Algorithm Convergence]
\label{thm:convergence}
Assume compact feasible sets for the projected blocks and a diminishing step-size sequence $\{\gamma_t\}$ satisfying the Robbins--Monro conditions. Then the sequences $\{\mathcal{X}^{(t)}\}$ and $\{\bm{\lambda}^{(t)}\}$ generated by Algorithm~\ref{alg:ml_enhanced_ssl} satisfy:

\emph{(a) Deterministic/full-CSI case.} Suppose $\nabla \mathcal{L}_{\mathrm{aug}}$ is blockwise Lipschitz in a neighborhood of the limit set. Then, for sufficiently large $t$:
\[
\mathcal{L}_{\mathrm{aug}}(\mathcal{X}^{(t+1)}, \bm{\lambda}^{(t)}) 
\le \mathcal{L}_{\mathrm{aug}}(\mathcal{X}^{(t)}, \bm{\lambda}^{(t)}) 
- \eta_t \, \|\mathcal{X}^{(t+1)}-\mathcal{X}^{(t)}\|_F^2,
\]
for some $\eta_t>0$ depending on local smoothness/curvature constants (monotone descent up to projection). Moreover,
\[
\begin{split}
\lim_{t\to\infty}\big\|\mathcal{G}_{\gamma_t}(\mathcal{X}^{(t)},\boldsymbol{\lambda}^{(t)})\big\|_F
&= 0,\\
\lim_{t\to\infty}\big\|\nabla_{\boldsymbol{\lambda}}
\mathcal{L}_{\mathrm{aug}}(\mathcal{X}^{(t)}, \boldsymbol{\lambda}^{(t)})\big\|_2
&= 0.
\end{split}
\]

and every accumulation point $(\mathcal{X}^\star,\bm{\lambda}^\star)$ satisfies the KKT conditions of the constrained problem.

\emph{(b) Stochastic/partial-CSI case.} Suppose the gradient estimators are unbiased with bounded variance, i.e., 
\[
\begin{split}
\EX[\widehat{\nabla}\mathcal{L}_{\mathrm{aug}}\,|\,\mathcal{X}^{(t)},\bm{\lambda}^{(t)}]
&=\nabla\mathcal{L}_{\mathrm{aug}},\\
\EX\big[\|\widehat{\nabla}\mathcal{L}_{\mathrm{aug}}-\nabla\mathcal{L}_{\mathrm{aug}}\|_F^2
\,\big|\,\mathcal{X}^{(t)},\bm{\lambda}^{(t)}\big]
&\le \sigma^2.
\end{split}
\]

Then
\[
\begin{aligned}
\EX\!\big[\mathcal{L}_{\mathrm{aug}}(&\mathcal{X}^{(t+1)}, \bm{\lambda}^{(t)}) 
\,\big|\, \mathcal{X}^{(t)},\bm{\lambda}^{(t)}\big]
\le \mathcal{L}_{\mathrm{aug}}(\mathcal{X}^{(t)}, \bm{\lambda}^{(t)}) \\
&\quad - c\,\gamma_t\,\EX\!\big[\|\mathcal{G}_{\gamma_t}(\mathcal{X}^{(t)},\bm{\lambda}^{(t)})\|_F^2\big]
+ O(\gamma_t^2),
\end{aligned}
\]
for some \(c>0\), and (by Robbins--Siegmund) almost surely
\[
\begin{split}
\liminf_{t\to\infty}\,
\EX\!\big[\|\mathcal{G}_{\gamma_t}(\mathcal{X}^{(t)},\boldsymbol{\lambda}^{(t)})\|_F^2\big]
&=0,\\
\lim_{t\to\infty}
\EX\!\big[\|\nabla_{\boldsymbol{\lambda}}
\mathcal{L}_{\mathrm{aug}}(\mathcal{X}^{(t)}, \boldsymbol{\lambda}^{(t)})\|_2\big]
&=0.
\end{split}
\]
Consequently, the sequence is asymptotically stationary in expectation, and any limit point is a stationary/KKT point with probability one under standard boundedness assumptions.
\end{theorem}

\begin{proof}
See Appendix~B.
\end{proof}

The above guarantees ensure that the alternating projected-gradient scheme with augmented Lagrangian handling of \eqref{eq:leakage_constraint}--\eqref{eq:power_constraint} converges to stationary points while maintaining practical descent behavior in the presence of non-convex couplings and unit-modulus projections. In the partial-CSI regime, monotone descent holds in \emph{expectation} (up to $O(\gamma_t^2)$ noise), and almost-sure stationarity follows from standard stochastic approximation arguments.

\begin{corollary}[Sublinear Rate]
\label{thm:convergence_rate}
\emph{(Deterministic/full-CSI).} Suppose the $\bm{W}$-subproblem is $\mu_W$-strongly convex (locally) and the $\bm{\Phi}$-subproblem has $L_{\bm{\Phi}}$-Lipschitz continuous gradients in a neighborhood of the limit set. Then the augmented Lagrangian satisfies
\begin{equation}
\label{eq:rate_convergence}
\mathcal{L}_{\mathrm{aug}}(\mathcal{X}^{(t)}, \bm{\lambda}^{(t)}) - \mathcal{L}_{\mathrm{aug}}^* \le \frac{C}{t},
\end{equation}
where
\[
C \;=\; \max\!\left\{ \frac{L_{\bm{\Phi}}^2}{2\mu_W}, \ \frac{L_{\bm{\Phi}}}{2} \right\} \,\|\mathcal{X}^{(0)}-\mathcal{X}^*\|_F^2,
\]
and $\mathcal{L}_{\mathrm{aug}}^*$ denotes the limiting value of the augmented Lagrangian.
\end{corollary}

\begin{proof}
Define the block functions at iteration $t$:
\begin{align}
f_1(\bm{W}) &= \mathcal{L}_{\mathrm{aug}}(\bm{W}, \bm{\Phi}^{(t)}, \bm{C}^{(t)}, \bm{\lambda}^{(t)}), \\
f_2(\bm{\Phi}) &= \mathcal{L}_{\mathrm{aug}}(\bm{W}^{(t+1)}, \bm{\Phi}, \bm{C}^{(t)}, \bm{\lambda}^{(t)}), \\
f_3(\bm{C}) &= \mathcal{L}_{\mathrm{aug}}(\bm{W}^{(t+1)}, \bm{\Phi}^{(t+1)}, \bm{C}, \bm{\lambda}^{(t)}).
\end{align}
By $\mu_W$-strong convexity of $f_1$ (locally) and Lipschitz continuity of $\nabla f_2$ with constant $L_{\bm{\Phi}}$, a standard descent lemma plus block coordinate arguments yield
\[
f_1(\bm{W}^{(t+1)}) - f_1(\bm{W}^*) \;\le\; \frac{1}{2\mu_W}\,\|\nabla f_1(\bm{W}^{(t+1)})\|_F^2,
\]
and
\[
f_2(\bm{\Phi}^{(t+1)}) - f_2(\bm{\Phi}^*) \;\le\; \frac{L_{\bm{\Phi}}}{2}\,\|\bm{\Phi}^{(t+1)}-\bm{\Phi}^*\|_F^2.
\]
The $\bm{C}$-block is minimized exactly (closed form per the MMSE update, cf.\ Section~\ref{sec:optimization}), so it does not increase the objective. Summing and telescoping the resulting inequalities over $t$, and using the diminishing $\{\gamma_t\}$ schedule, gives \eqref{eq:rate_convergence}.
\end{proof}

\begin{corollary}[Accelerated Sublinear Rate]
\emph{(Deterministic/full-CSI).} Under the conditions of Corollary~\ref{thm:convergence_rate}, if, in addition, the $\bm{C}$-subproblem is $\mu_C$-strongly convex (locally) and one employs Nesterov-style extrapolation on each block, then
\begin{equation}
\mathcal{L}_{\mathrm{aug}}(\mathcal{X}^{(t)}, \bm{\lambda}^{(t)}) - \mathcal{L}_{\mathrm{aug}}^* \;\le\; \frac{C'}{t^2},
\end{equation}
where
\[
C' \;=\; \max\{\mu_W^{-1},\, L_{\bm{\Phi}}/\mu_W,\, \mu_C^{-1}\}\; L^2 \, \|\mathcal{X}^{(0)}-\mathcal{X}^*\|_F^2,
\]
and $L$ is a global Lipschitz constant for $\nabla \mathcal{L}_{\mathrm{aug}}$ over the considered neighborhood.
\end{corollary}

\begin{proof}
Introduce the quadratic surrogates per block with local curvature parameters $\mu_i$:
\[
Q_i(\bm{x}, \bm{y}) \;=\; \mathcal{L}_{\mathrm{aug}}(\bm{y}) \;+\; \langle \nabla \mathcal{L}_{\mathrm{aug}}(\bm{y}), \bm{x}-\bm{y} \rangle \;+\; \frac{\mu_i}{2}\,\|\bm{x}-\bm{y}\|_F^2.
\]
For block $i\in\{ \bm{W},\bm{\Phi},\bm{C}\}$, use the accelerated updates
\begin{align}
\bm{y}_i^{(k)} &= \bm{x}_i^{(k)} + \frac{k-1}{k+2}\big(\bm{x}_i^{(k)}-\bm{x}_i^{(k-1)}\big), \\
\bm{x}_i^{(k+1)} &= \arg\min_{\bm{x}} Q_i(\bm{x}, \bm{y}_i^{(k)}),
\end{align}
together with projection $\mathcal{P}_{\mathcal{U}}$ for the $\bm{\Phi}$-block. Standard accelerated block-coordinate analysis (with Lipschitz gradients and block strong convexity) yields the $O(1/t^2)$ rate.
\end{proof}

\section{Unknown CSI: Bayesian Optimization and Learning}
\label{sec:bayesian}

The secure transmission framework developed in the previous sections relied on known or reliably estimated CSI, leveraging either closed-form designs or alternating optimization schemes. However, in practical IRS-assisted communication systems, accurate CSI acquisition is often hindered by estimation overhead, hardware limitations, and user mobility.

To address this challenge, we adopt a Bayesian optimization framework that explicitly accounts for channel uncertainty. Rather than assuming access to deterministic or perfectly estimated CSI, we model the composite channel as a random variable:
\begin{equation}
\mathcal{H} \;=\; \Big\{ \bm{H}, \{ \bm{h}_k \}_{k=1}^{K_H}, \{ \bm{g}_e \}_{e=1}^{K_B} \Big\},
\end{equation}
where the joint distribution $p(\mathcal{H})$ is unknown but assumed stationary over the design horizon. This stochastic modeling transforms the original design into one that optimizes the \emph{expected} performance over possible channel realizations.

The composite channel $\mathcal{H}$ consists of $\bm{H} \in \mathbb{C}^{N \times M}$, $\bm{h}_k \in \mathbb{C}^{N \times 1}$, and $\bm{g}_e \in \mathbb{C}^{N \times 1}$, representing the BS--IRS channel matrix, the IRS--user-$k$ channel vector, and the IRS--eavesdropper-$e$ channel vector, respectively. Although the exact form of $p(\mathcal{H})$ is unknown, robust designs can be obtained by querying system performance under sampled realizations, without relying on explicit CSI estimation.

Let $\mathcal{X}$ denote the full design space
\begin{equation}
\begin{split}
  \mathcal{X} \equiv \big\{(\bm{\Phi},\bm{W},\bm{C}) :\;&
    |[\bm{\Phi}]_{n,n}| = 1,\\
  &\bm{W} \in \mathbb{C}^{M \times K_H},\\
  &\bm{C} = \mathrm{diag}(c_1,\ldots,c_{K_H})
  \big\}.
\end{split}
\end{equation}

To reduce computational complexity, we embed this structured space into a lower-dimensional Euclidean latent domain $\mathcal{Z} \subset \mathbb{R}^d$ via an encoder $\psi:\mathcal{X}\to\mathcal{Z}$ and deploy designs via a decoder $\varphi:\mathcal{Z}\to\mathcal{X}$. The decoder accounts for the fact that general random projections are not invertible.

\begin{theorem}[Finite-set near-isometry via random projection]
There exists a feature map $\phi:\mathcal{X}\to\mathbb{R}^{D}$ defined by
\begin{equation}
\begin{split}
\phi(\bm{\Phi},\bm{W},\bm{C})
&=
\Big[
\underbrace{\mathrm{Re}\,\mathrm{diag}(\bm{\Phi})\,;\ 
\mathrm{Im}\,\mathrm{diag}(\bm{\Phi})}_{\text{phase embedding }(2N)}\\
&\quad;\ 
\mathrm{vec}(\bm{W})\ ;\ 
\mathrm{Re}\,\mathrm{diag}(\bm{C})\ ;\ 
\mathrm{Im}\,\mathrm{diag}(\bm{C})
\Big],
\end{split}
\end{equation}

such that for any finite subset $\mathcal{S}\subset\mathcal{X}$ of cardinality $T$ and any $\epsilon\in(0,1)$ there exists a random linear map $\bm{R}:\mathbb{R}^{D}\to\mathbb{R}^{d}$ with
\[
d = \mathcal{O}\!\big(\epsilon^{-2}\log T\big)
\]
satisfying, with high probability, for all $\bm{x}_1,\bm{x}_2\in\mathcal{S}$,
\begin{equation}
\begin{split}
(1-\epsilon)\,\|\phi(\bm{x}_1) - \phi(\bm{x}_2)\|_2
&\le
\|\bm{R}\phi(\bm{x}_1) - \bm{R}\phi(\bm{x}_2)\|_2\\
&\le
(1+\epsilon)\,\|\phi(\bm{x}_1) - \phi(\bm{x}_2)\|_2.
\end{split}
\end{equation}

In particular, one may set $\psi=\bm{R}\circ\phi$. A decoder $\varphi$ (e.g., learned or structured) provides a reconstruction $\varphi(\bm{z})\approx\bm{x}$ for deployment. \hfill\mbox{}
\end{theorem}

\begin{proof}
Apply the Johnson--Lindenstrauss (JL) lemma to the finite set $\phi(\mathcal{S})$~\cite{johnson1984extensions}. The phase embedding uses the real/imaginary parts of $\mathrm{diag}(\bm{\Phi})$ (points on the unit circle), which preserves Euclidean distances in $\mathbb{R}^{2N}$. Composition with a JL projection yields the stated near-isometry; see also random-embedding BO guarantees~\cite{wang2016bayesian}.
\end{proof}

The reduced optimization problem becomes
\begin{equation}
\label{eq:bo_stochastic_prog}
\begin{split}
\min_{\bm{z}\in\mathcal{Z}}\; & 
\mathbb{E}_{\mathcal{H}}\Big[
\sum_{k=1}^{K_H} \mathrm{MSE}_k\big(\varphi(\bm{z})\big) \\
&\qquad\qquad
+ \lambda \sum_{e=1}^{K_B} P_e^{(\mathrm{sig})}\big(\varphi(\bm{z})\big) 
+ \mu \big\|\bm{W}\big(\varphi(\bm{z})\big)\big\|_F^2 
\Big] \\[0.5ex]
\text{s.t.}\quad 
& \mathbb{E}_{\mathcal{H}}\big[ \big\|\bm{W}\big(\varphi(\bm{z})\big)\big\|_F^2\big] \le P_{\max},\\
& \mathbb{E}_{\mathcal{H}}\big[ \sum_{e=1}^{K_B} P_e^{(\mathrm{sig})}\big(\varphi(\bm{z})\big)\big] \le \Gamma_{\mathrm{leak}}.
\end{split}
\end{equation}
Here, $\varphi(\bm{z})=(\bm{\Phi}(\bm{z}),\bm{W}(\bm{z}),\bm{C}(\bm{z}))$ is the decoded design used for on-device evaluation; the equalizer matrix $\bm{C}$ is part of the decision and implicitly appears in each evaluation.

To respect the product structure of $\mathcal{X}$, we employ a separable kernel over blocks on the \emph{decoded} inputs:
\begin{equation}
\begin{split}
k_{\mathcal{X}}\!\bigl((\bm{\Phi},\bm{W},\bm{C}),(\bm{\Phi}',\bm{W}',\bm{C}')\bigr)
&= k_{\mathrm{torus}}(\bm{\Phi},\bm{\Phi}')\\
&\quad\cdot\,k_W(\bm{W},\bm{W}')\\
&\quad\cdot\,k_C(\bm{C},\bm{C}').
\end{split}
\end{equation}
with $k_{\mathrm{torus}}$ a periodic kernel on diagonal phases (e.g., via wrapped distances), and $k_W$, $k_C$ standard kernels on Euclidean blocks (e.g., automatic relevance determination radial basis function (ARD-RBF) kernel on vectorization). In practice we use the latent-space kernel
\[
k_{\mathcal{Z}}(\bm{z},\bm{z}')
\;=\;
k_{\mathcal{X}}\!\big(\varphi(\bm{z}),\,\varphi(\bm{z}')\big),
\]
which is geometry-aware through the decoder. The GP surrogate is trained on pairs $\big(\bm{z}_i,\,y_i\big)$ while evaluating kernels via the decoded designs.

\subsection{Gaussian Process Surrogate Modeling}
\label{subsec:gp_modeling}

To reduce the cost of evaluating the black-box objective under uncertain channel conditions, we employ GP regression as a surrogate model in the latent space $\mathcal{Z} \subset \mathbb{R}^d$. At iteration $t$, let $\mathcal{D}_t = \{(\bm{z}_i, y_i)\}_{i=1}^t$ denote the dataset of latent design points and their corresponding noisy evaluations.

The true objective is estimated via Monte Carlo averaging over channel realizations:
\begin{equation}
\label{eq:objective_mc}
\begin{split}
f(\bm{z}) \;=\; 
& \frac{1}{N_{\mathrm{MC}}} \sum_{n=1}^{N_{\mathrm{MC}}} \Bigg(
\sum_{k=1}^{K_H} \mathrm{MSE}_k^{(n)}\!\big(\varphi(\bm{z})\big) 
+ \lambda \sum_{e=1}^{K_B} P_{e,\mathrm{sig}}^{(n)}\!\big(\varphi(\bm{z})\big) \\[-0.25ex]
&\qquad\qquad\qquad\qquad
+ \mu \left\| \bm{W}^{(n)}\!\big(\varphi(\bm{z})\big) \right\|_F^2
\Bigg),
\end{split}
\end{equation}
where $\lambda>0$ balances MSE and secrecy pressure, and the superscript $(n)$ denotes the $n$-th Monte Carlo draw of $\mathcal{H}$. (Equivalently, $P_{e,\mathrm{sig}}^{(n)}(\cdot)$ is the Monte-Carlo instance of $P_e^{(\mathrm{sig})}(\cdot)$.)

We place a zero-mean GP prior on $f$ with covariance $k_{\mathcal{Z}}$:
\begin{equation}
f \sim \mathcal{GP}\big(0,\, k_{\mathcal{Z}}(\cdot,\cdot)\big).
\end{equation}
Given noisy observations $y_i = f(\bm{z}_i) + \varepsilon_i$ with $\varepsilon_i \sim \mathcal{N}(0,\sigma_n^2)$ (capturing measurement and Monte-Carlo noise), the GP posterior at a test point $\bm{z}$ admits closed-form expressions.

\begin{proposition}
\label{prop:gp_prediction}
Let $\bm{K}_t \in \mathbb{R}^{t \times t}$ be the kernel matrix with $[\bm{K}_t]_{ij} = k_{\mathcal{Z}}(\bm{z}_i,\bm{z}_j)$, and let $\bm{k}_t(\bm{z}) \in \mathbb{R}^t$ with $[\bm{k}_t(\bm{z})]_i = k_{\mathcal{Z}}(\bm{z},\bm{z}_i)$. Then
\begin{equation}
f(\bm{z}) \mid \mathcal{D}_t \sim \mathcal{N}\!\left( \mu_t(\bm{z}),\, \sigma_t^2(\bm{z}) \right),
\end{equation}
with predictive mean and variance
\begin{align}
\mu_t(\bm{z}) &= \bm{k}_t^\top(\bm{z}) \big( \bm{K}_t + \sigma_n^2 \bm{I} \big)^{-1} \bm{y}_t, \label{eq:gp_mean} \\
\sigma_t^2(\bm{z}) &= k_{\mathcal{Z}}(\bm{z},\bm{z}) - \bm{k}_t^\top(\bm{z}) \big( \bm{K}_t + \sigma_n^2 \bm{I} \big)^{-1} \bm{k}_t(\bm{z}). \label{eq:gp_var}
\end{align}
\end{proposition}

\begin{proof}
This follows from conditioning properties of multivariate Gaussian distributions. The predictive variance decreases as informative data are acquired and vanishes whenever $\bm{z}$ lies in the span induced by $\{\bm{z}_i\}_{i=1}^t$ under the kernel.
\end{proof}

This GP surrogate enables sequential Bayesian optimization via acquisition functions. The overall constrained Bayesian optimization procedure is summarized in Algorithm~\ref{alg:constrained_bo}.

\begin{algorithm}[t]
\caption{Constrained Bayesian Optimization under Channel Uncertainty}
\label{alg:constrained_bo}
\SetAlFnt{\small}
\SetAlCapFnt{\small}
\SetAlCapNameFnt{\small}
\DontPrintSemicolon
\setlength{\abovecaptionskip}{4pt}
\setlength{\belowcaptionskip}{2pt}

\KwIn{Latent space $\mathcal{Z}\!\subset\!\mathbb{R}^d$, GP kernel $k_{\mathcal{Z}}$, budget $P_{\max}$, leakage $\Gamma_{\mathrm{leak}}$, MC samples $N_{\mathrm{MC}}$}
\KwOut{Optimal design $(\bm{\Phi}^*,\bm{W}^*,\bm{C}^*)=\varphi(\bm{z}^*)$}

\textbf{Init:} $\mathcal{D}_0\!\gets\!\emptyset$. Sample $\{\bm{z}_i\}_{i=1}^{N_{\mathrm{init}}}\subset\mathcal{Z}$ via Latin hypercube.\\
\For{$i=1$ \KwTo $N_{\mathrm{init}}$}{
  $(\bm{\Phi}_i,\bm{W}_i,\bm{C}_i)\!\gets\!\varphi(\bm{z}_i)$; \\
  Evaluate
  \[
  \begin{split}
    f_i &= \frac{1}{N_{\mathrm{MC}}}\sum_{n=1}^{N_{\mathrm{MC}}}\Biggl(
      \sum_{k=1}^{K_H} \mathrm{MSE}_k^{(n)}(\varphi(\bm{z}_i))
      + \lambda \sum_{e=1}^{K_B} P_{e,\mathrm{sig}}^{(n)}(\varphi(\bm{z}_i))\\
    &\qquad\qquad\qquad\qquad
      + \mu \big\| \bm{W}^{(n)}(\varphi(\bm{z}_i)) \big\|_F^2
    \Biggr),\\
    g_{1,i} &= \frac{1}{N_{\mathrm{MC}}} \sum_{n=1}^{N_{\mathrm{MC}}} \|\bm{W}^{(n)}(\varphi(\bm{z}_i))\|_F^2 - P_{\max},\\
    g_{2,i} &= \frac{1}{N_{\mathrm{MC}}} \sum_{n=1}^{N_{\mathrm{MC}}} \sum_{e=1}^{K_B} P_{e,\mathrm{sig}}^{(n)}(\varphi(\bm{z}_i)) - \Gamma_{\mathrm{leak}}.
  \end{split}
  \]
  Update $\mathcal{D}_0 \gets \mathcal{D}_0 \cup \{(\bm{z}_i,f_i,g_{1,i},g_{2,i})\}$;
}

$t\!\gets\!1$\\
\While{not converged and $t<T$}{
  Train GPs for $f,g_1,g_2$:
  \[
    \bm{\theta}_j^* = \arg\max_{\bm{\theta}} \log p(\mathcal{D}_{t-1}^{(j)}\mid \bm{\theta}),\quad j\in\{f,g_1,g_2\}
  \]
  Select
  \[
    \bm{z}_t = \arg\max_{\bm{z}\in\mathcal{Z}} \mathrm{EI}_t(\bm{z}) \cdot \prod_{j=1}^2 \mathbb{P}(g_j(\bm{z})\le 0)
  \]
  $(\bm{\Phi}_t,\bm{W}_t,\bm{C}_t)\!\gets\!\varphi(\bm{z}_t)$; evaluate $f_t,g_{1,t},g_{2,t}$ as above.\\
  $\mathcal{D}_t \gets \mathcal{D}_{t-1} \cup \{(\bm{z}_t,f_t,g_{1,t},g_{2,t})\}$;\\
  $t\!\gets\!t+1$
}

\textbf{Post-process:} Define feasible set
\[
\mathcal{F}_T = \{ \bm{z}\in\mathcal{D}_T \mid g_1(\bm{z})\le 0,\ g_2(\bm{z})\le 0\}.
\]
\textbf{Output:} $\bm{z}^* = \arg\min_{\bm{z}\in\mathcal{F}_T} f(\bm{z})$, and return $\varphi(\bm{z}^*)$.
\end{algorithm}

\subsection{Convergence Guarantee}
\label{sec:bayesian_convergence}

This section establishes the convergence behavior of our constrained Bayesian optimization algorithm (Algorithm~\ref{alg:constrained_bo}) applied to the stochastic formulation in \eqref{eq:bo_stochastic_prog}. The procedure operates in a latent design space $\mathcal{Z}$, where each candidate $\bm{z} \in \mathcal{Z}$ is mapped to a full system configuration via a decoder $\varphi: \mathcal{Z} \rightarrow \mathcal{X}$. The objective is to iteratively minimize a black-box performance metric under probabilistic constraints.

\begin{theorem}[Convergence Guarantee]
\label{thm:bo_convergence}
Let $f: \mathcal{Z} \to \mathbb{R}$ be the latent objective defined as $f(\bm{z}) := F(\varphi(\bm{z}))$, where $f \sim \mathcal{GP}(0, k_{\mathcal{Z}})$ and $\sup_{\bm{z}} k_{\mathcal{Z}}(\bm{z}, \bm{z}) \leq 1$. Given noisy evaluations $y_t = f(\bm{z}_t) + \epsilon_t$ with $\epsilon_t \sim \mathcal{N}(0, \sigma^2)$ and a query sequence $\{\bm{z}_t\}_{t=1}^T$ generated by the constrained BO algorithm, define the cumulative regret
\begin{equation}
R_T = \sum_{t=1}^T \big[ f(\bm{z}_t) - f(\bm{z}^*) \big],
\end{equation}
where $\bm{z}^* = \arg\min_{\bm{z} \in \mathcal{F}_\delta} f(\bm{z})$ and $\mathcal{F}_\delta = \{ \bm{z} \mid \mathbb{P}(g_j(\bm{z}) \leq 0) \geq 1-\delta,~\forall j \}$.

Then, with probability at least $1 - \delta$, the regret is bounded as
\begin{equation}
R_T \;\leq\; \sqrt{C_1\, T\, \beta_T\, \gamma_T} \;+\; C_2 \sqrt{T \log(1/\delta)},
\end{equation}
where $\gamma_T$ is the maximum information gain,
\begin{equation}
\gamma_T \;=\; \max_{\{\bm{z}_t\}} 
\frac{1}{2} \log \det \!\left( \bm{I}_T + \frac{1}{\sigma^2} \bm{K}_T \right),
\end{equation}
and $\beta_T = 2 \|f\|_{\mathcal{H}_{k_{\mathcal{Z}}}}^2 + 300\, \gamma_T \log^3(T/\delta)$ with reproducing kernel Hilbert space (RKHS) norm $\|f\|_{\mathcal{H}_{k_{\mathcal{Z}}}}$. Constants $C_1, C_2 > 0$ depend on $k_{\mathcal{Z}}$ and $\sigma^2$. Thus, the algorithm achieves no-regret:
\begin{equation}
\lim_{T \to \infty} \frac{R_T}{T} = 0 \quad \text{a.s.}
\end{equation}
\end{theorem}

\begin{proof}
See Appendix C.
\end{proof}

This result shows the latent-space Bayesian optimization framework attains sublinear regret compared to the best high-confidence feasible solution. As the surrogates for $g_j$ improve, the empirical feasibility estimates $\mathbb{P}(g_j(\boldsymbol{z})\le0)$ become more accurate, steering iterates into the high-probability feasible region while also improving the objective.

\section{Simulation Results}  
\label{sec:sim}  

In this section, we present numerical results validating the proposed secure learning–based opportunistic transmission framework.

\subsection{Simulation Setup}

We assume a BS with \(M=8\) antennas and an IRS with \(N=64\) elements. PU activity follows a binary Markov model with \(P_{01}=0.2\), \(P_{10}=0.3\) and a uniform prior. All channels are i.i.d.\ Rayleigh fading, \(\mathcal{CN}(0,1)\), and SU receiver noise is circularly symmetric complex Gaussian with unit variance. IRS phase shifts are initialized as \(\theta_i\sim\mathcal{U}(0,2\pi)\) with unit amplitude.

An alternating optimization updates the BS beamformer, IRS phases, and SU filters block-wise under known or partial CSI; when CSI is unknown, a trust-region Bayesian optimization subroutine—with a squared-exponential kernel and GP-UCB acquisition (\(\beta_{t+1}=0.4\log(2t+2)\))—guides decisions without gradients. Results are averaged over \(10^4\) Monte Carlo trials, with \(T_{\max}=100\) and convergence declared once the objective improves by less than \(10^{-3}\) for three consecutive iterations.

\begin{figure}[!t]
    \centering
    \begin{subfigure}[b]{0.48\columnwidth}
        \centering
        \includegraphics[width=\linewidth]{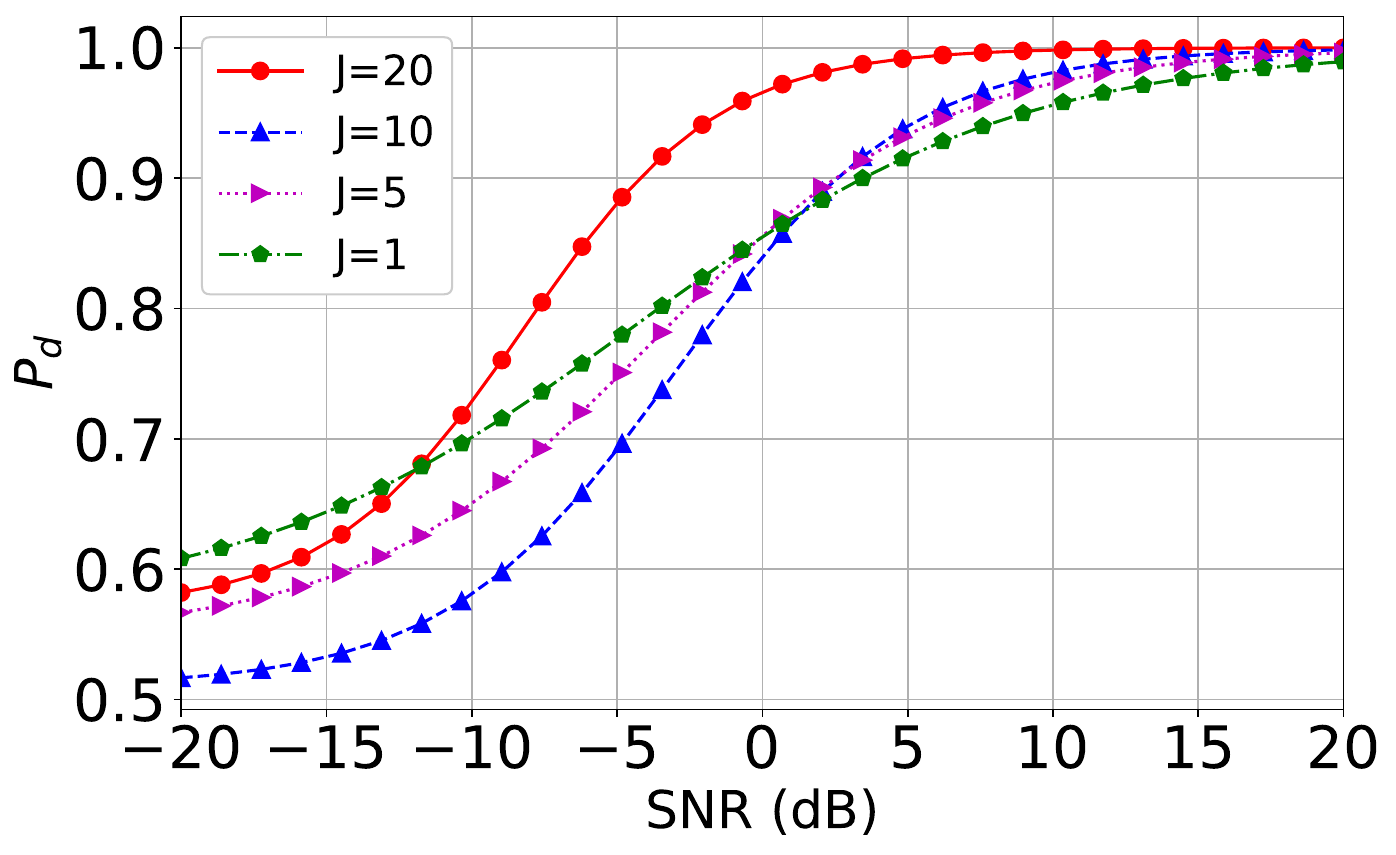}
        \caption{$K_H=8$}
        \label{fig:pdsnr_K8}
    \end{subfigure}
    \hfill
    \begin{subfigure}[b]{0.48\columnwidth}
        \centering
        \includegraphics[width=\linewidth]{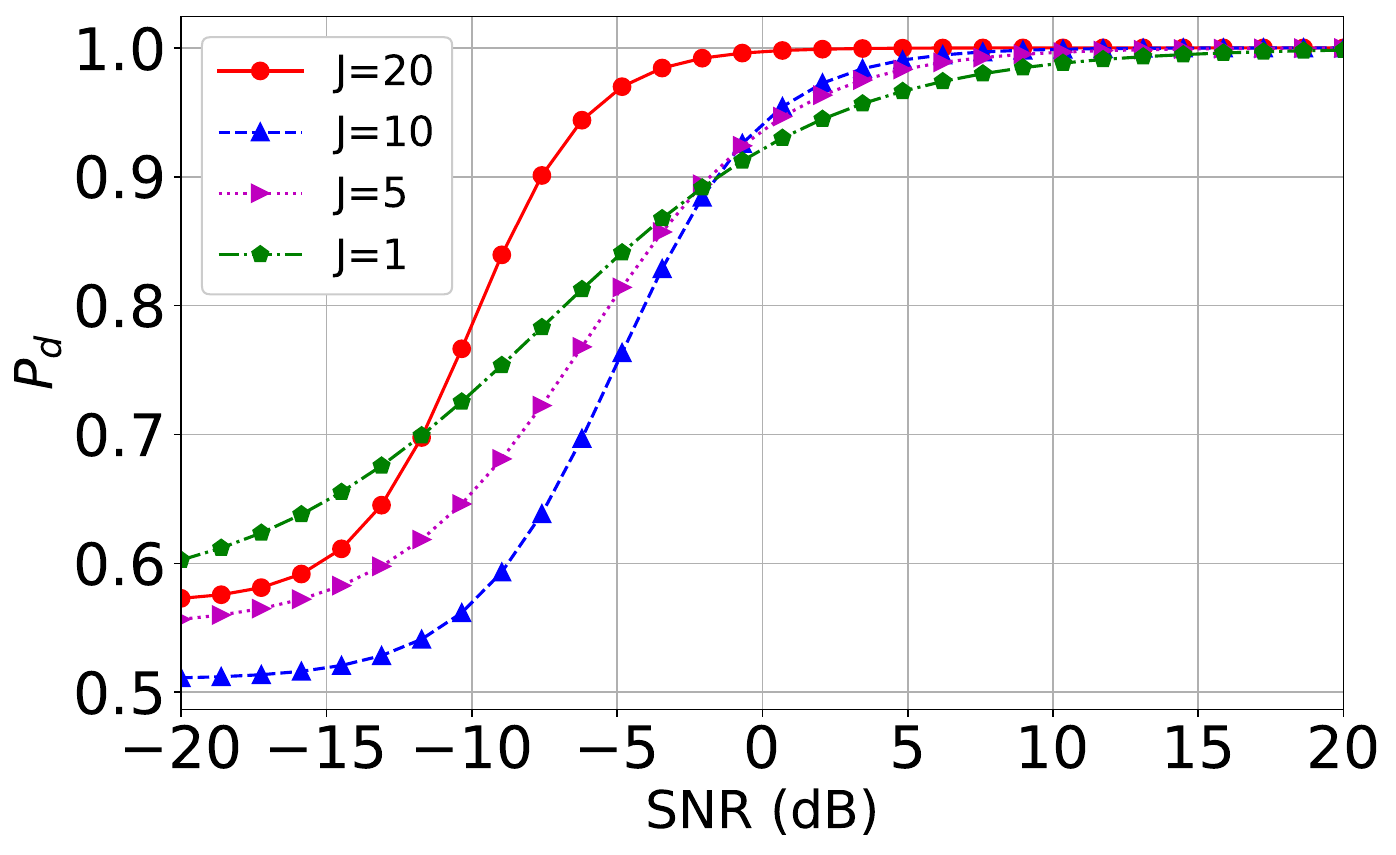}
        \caption{$K_H=16$}
        \label{fig:pdsnr_K16}
    \end{subfigure}
    \caption{Probability of detection $P_d$ versus SNR for different numbers of samples $J$ under a fixed false alarm rate $P_{\rm fa}=0.01$ and $\mathcal{R}=25$. (a) $K_H=8$, (b) $K_H=16$, with $K_B=3$ in both cases.}
    \label{fig:pdsnr}
\end{figure}
\subsection{Results and Discussions}
Fig.~\ref{fig:pdsnr} plots the detection probability \(P_d\) versus SNR for various sample sizes \(J\) at fixed \(P_{\rm fa}=0.01\) and maximum decision rounds \(\mathcal{R}=25\), under two configurations: (a) \(K_H=8\) and (b) \(K_H=16\), each with \(K_B=3\) Byzantine users.

The proposed algorithm remains robust even in low-SNR regimes, and increasing \(J\) significantly improves \(P_d\) by reducing estimation uncertainty. Although adding honest users (\(K_H\)) enhances performance in moderate-SNR settings, its marginal benefit diminishes when \(J\) is large due to accuracy saturation; nevertheless, a larger \(K_H\) accelerates convergence (see Fig.~\ref{fig:pditer}).

These results highlight a trade-off between communication (or energy) cost and statistical accuracy: strong detection performance is achieved even with moderate \(J\), making the method well-suited to resource-constrained sensing or transmission scenarios.

\begin{figure}[!t]
    \centering
    \includegraphics[scale=0.20]{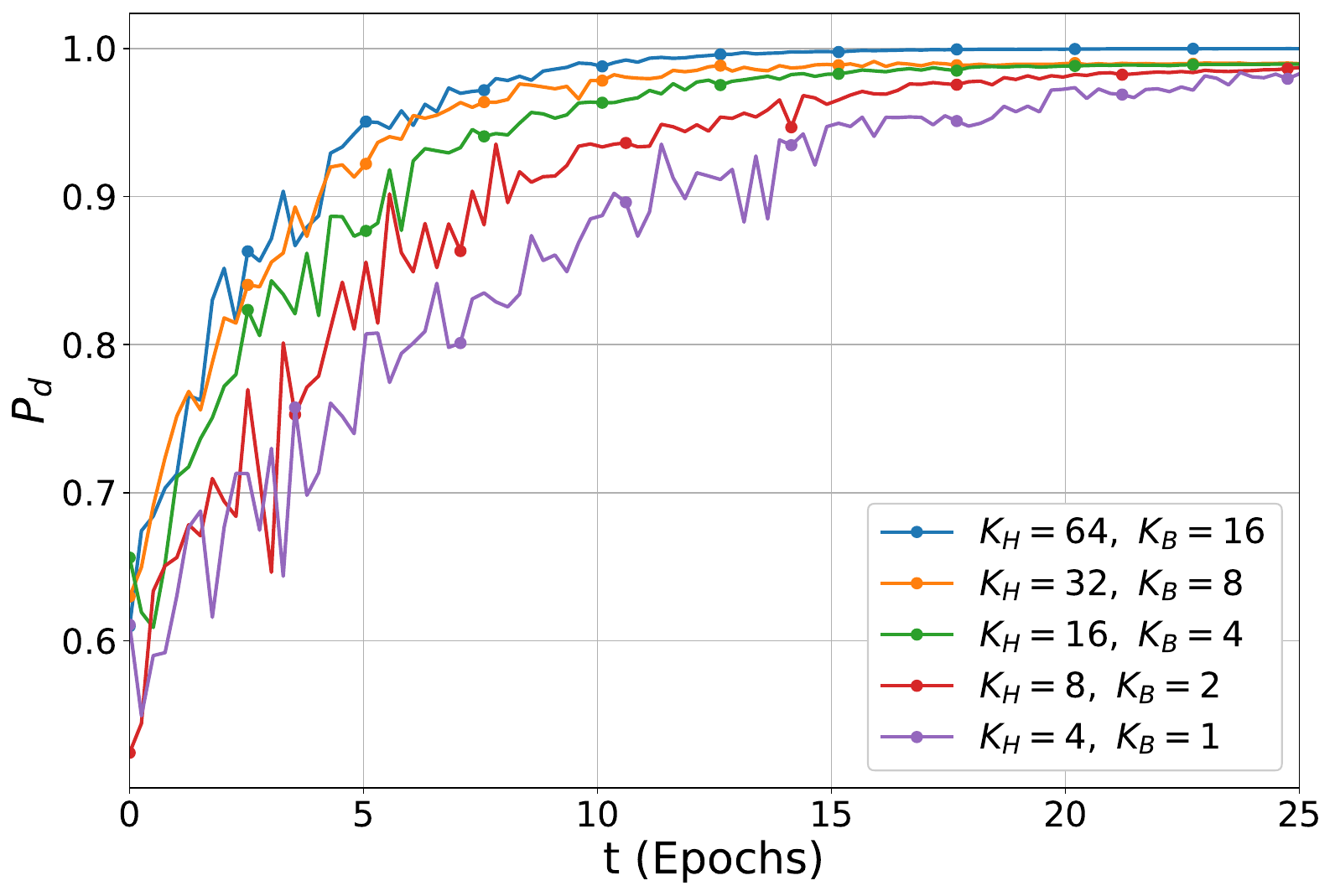}
    \caption{Convergence behavior of the proposed algorithm versus iteration index for $J=10$, $P_{\rm fa}=0.01$, and maximum iteration $\mathcal{R}=25$, under different user settings with $K_B/K_H = 1/4$.}
    \label{fig:pditer}
\end{figure}

Fig.~\ref{fig:pditer} plots the detection probability \(P_d\) versus iteration index for \(J=10\), \(P_{\rm fa}=0.01\), \(\mathcal{R}=25\), and a Byzantine-to-honest user ratio \(K_B/K_H=1/4\). The algorithm converges rapidly—stabilizing \(P_d\) within a few iterations—making it well-suited for real-time or delay-sensitive applications.

This fast convergence arises from distributed learning, whereby users collaboratively refine local decisions through iterative updates. As \(K_H\) increases, the final \(P_d\) improves but with diminishing returns, although additional honest users consistently speed up convergence.

These results demonstrate that the proposed method is computationally efficient, scalable, and maintains robust detection performance even in adversarial settings.

\begin{figure}[!t]
    \centering
    \subfloat[Alternating Learning]{\includegraphics[width=0.48\columnwidth]{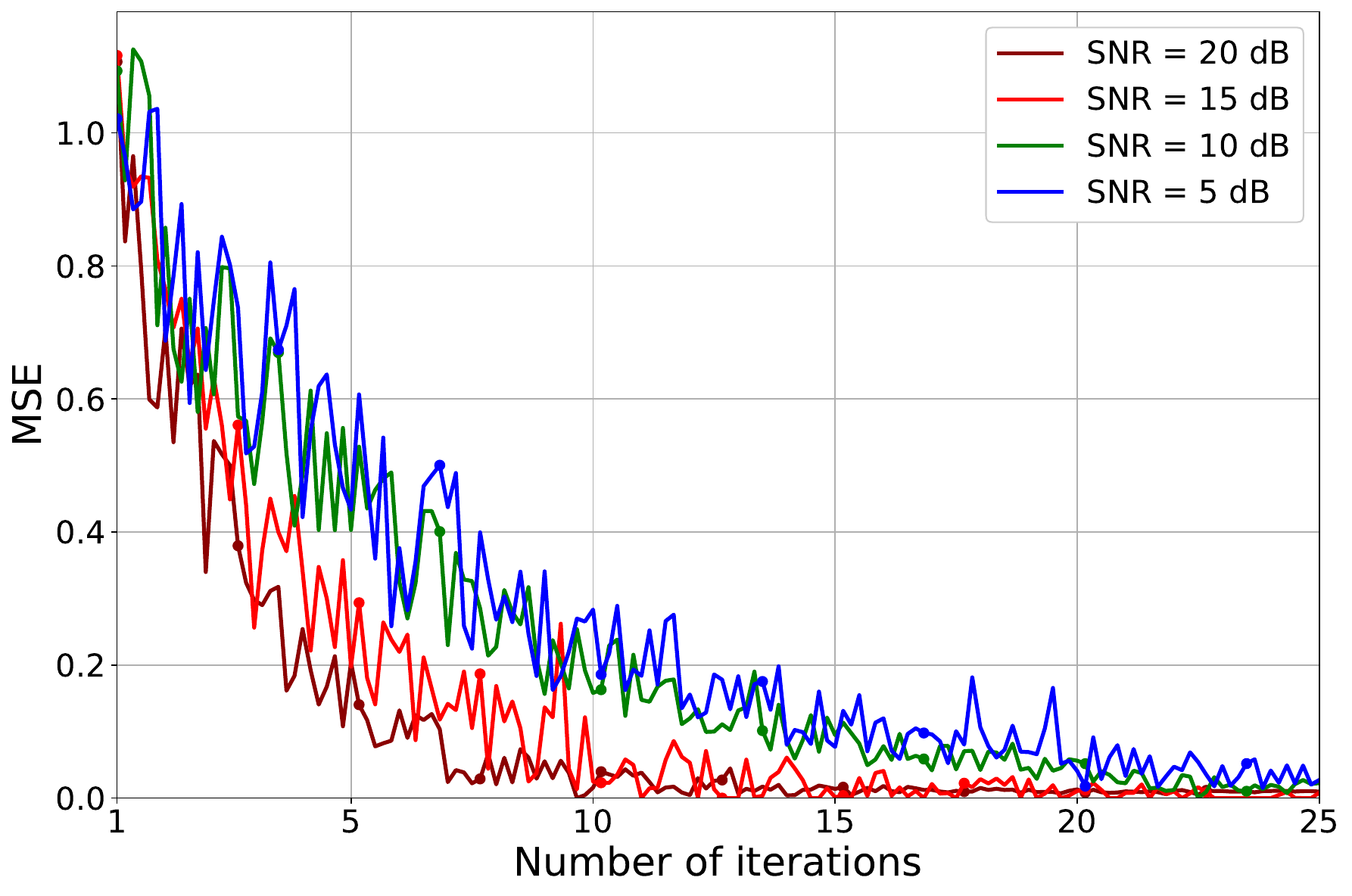}%
    \label{fig:mseiter_alt}} 
    \hfill
    \subfloat[Bayesian Learning]{\includegraphics[width=0.48\columnwidth]{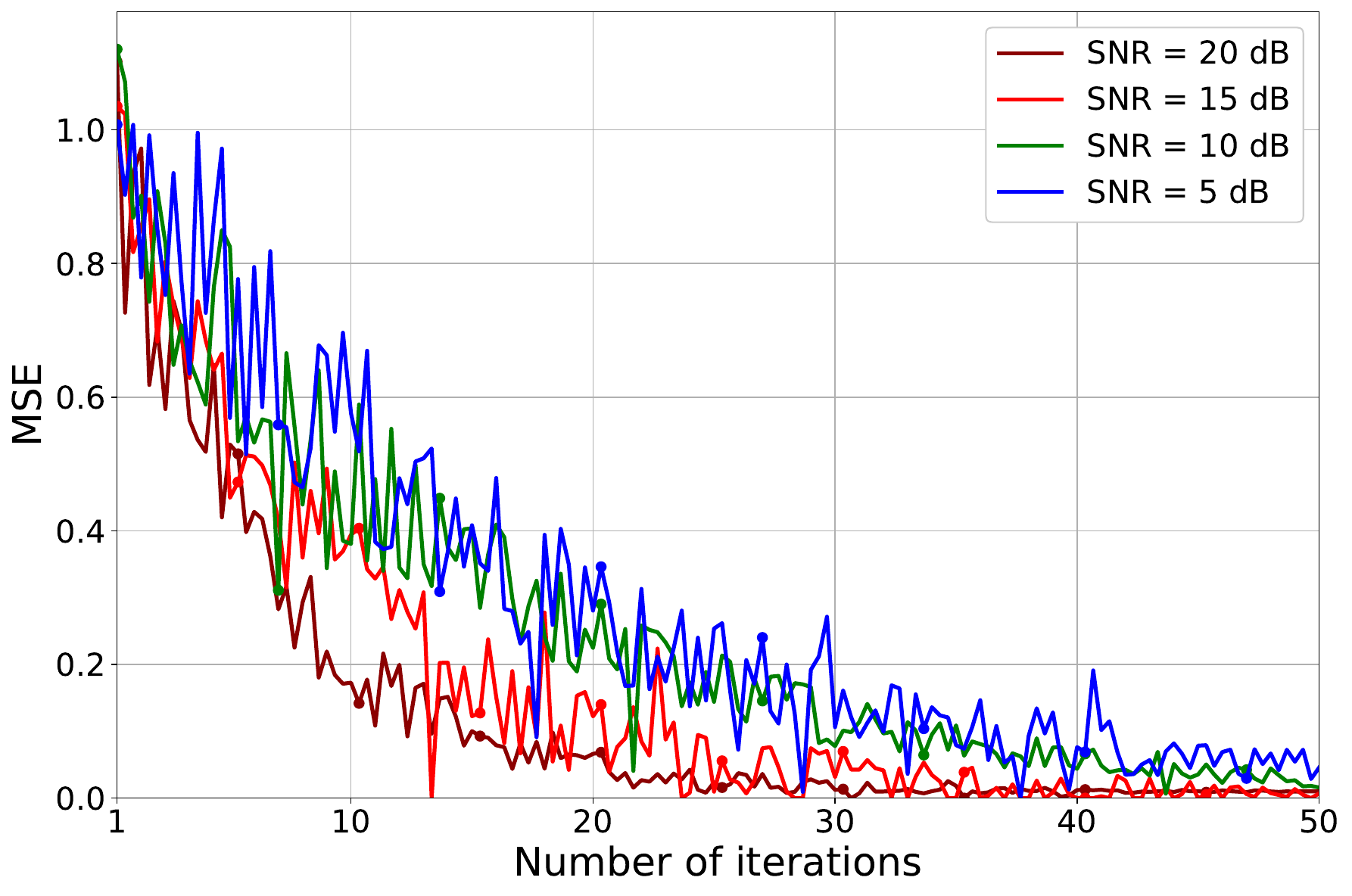}%
    \label{fig:mseiter_bayes}} 
    \caption{Convergence rate of MSE for the two proposed optimization approaches: (a) Alternating learning with partial CSI, and (b) Bayesian learning without CSI knowledge.}
    \label{fig:mseiter}
\end{figure}

Fig.~\ref{fig:mseiter} compares the convergence of the two optimization schemes: (a) alternating learning and (b) Bayesian learning. The alternating algorithm converges markedly faster—leveraging partial CSI to perform direct gradient updates on \(\{\mathbf{W},\boldsymbol{\Phi},\mathbf{C}\}\)—while the CSI‐free Bayesian approach relies on an acquisition function and thus needs more iterations. Nonetheless, Bayesian learning reaches competitive performance given sufficient iterations, confirming alternating learning as the choice when partial CSI is available and Bayesian learning as a robust alternative otherwise.

\begin{figure}[!t]
    \centering
    \includegraphics[scale=0.22]{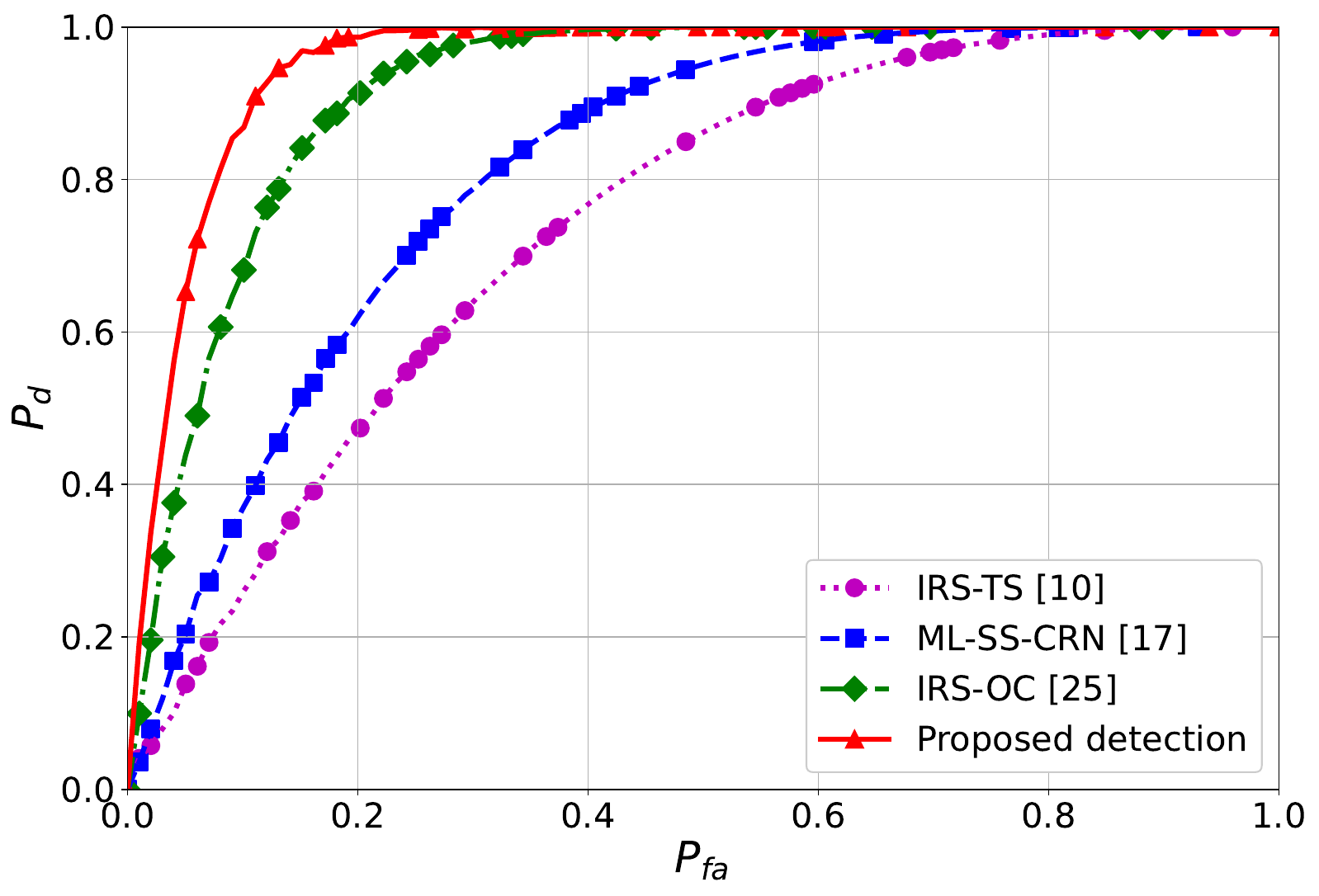}
    \caption{Receiver operating characteristic (ROC) curves comparing the proposed detection learning algorithm with benchmark methods: IRS-TS~\cite{Shao2022}, ML-SS-CRN~\cite{BB69}, and IRS-OC~\cite{Yashvanth2023}. Evaluation parameters: $\text{SNR}=5\,\text{dB}$, $J=10$ observations per user, $K_H=8$ honest users, $K_B=3$ adversarial nodes, and $\mathcal{R}=25$ resource units.}
    \label{fig:roc}
\end{figure}

Fig.~\ref{fig:roc} plots receiver operating characteristic (ROC) curves at \(\mathrm{SNR}=5\)\,dB, \(J=10\), \(K_H=8\), \(K_B=3\), \(\mathcal{R}=25\), comparing our distributed detection learning algorithm with IRS-TS~\cite{Shao2022}, ML-SS-CRN~\cite{BB69}, and IRS-OC~\cite{Yashvanth2023}. Our method consistently achieves higher \(P_d\) for all \(P_{\rm fa}\).

This gain stems from Bayesian belief updates that handle uncertainty, structured exploration for optimal sensing, and resilience to adversaries, enabling robust performance with limited feedback and in dynamic networks. When fixing \(K_HJ=80\) (e.g.\ \(K_H=16,J=5\)), the gap narrows but our approach retains an edge under mobility and topology changes. The ROC curves further demonstrate tunable operating points, allowing sensitivity–specificity trade-offs for practical deployment.

\begin{figure}[!t]
    \centering
    \includegraphics[scale=.22]{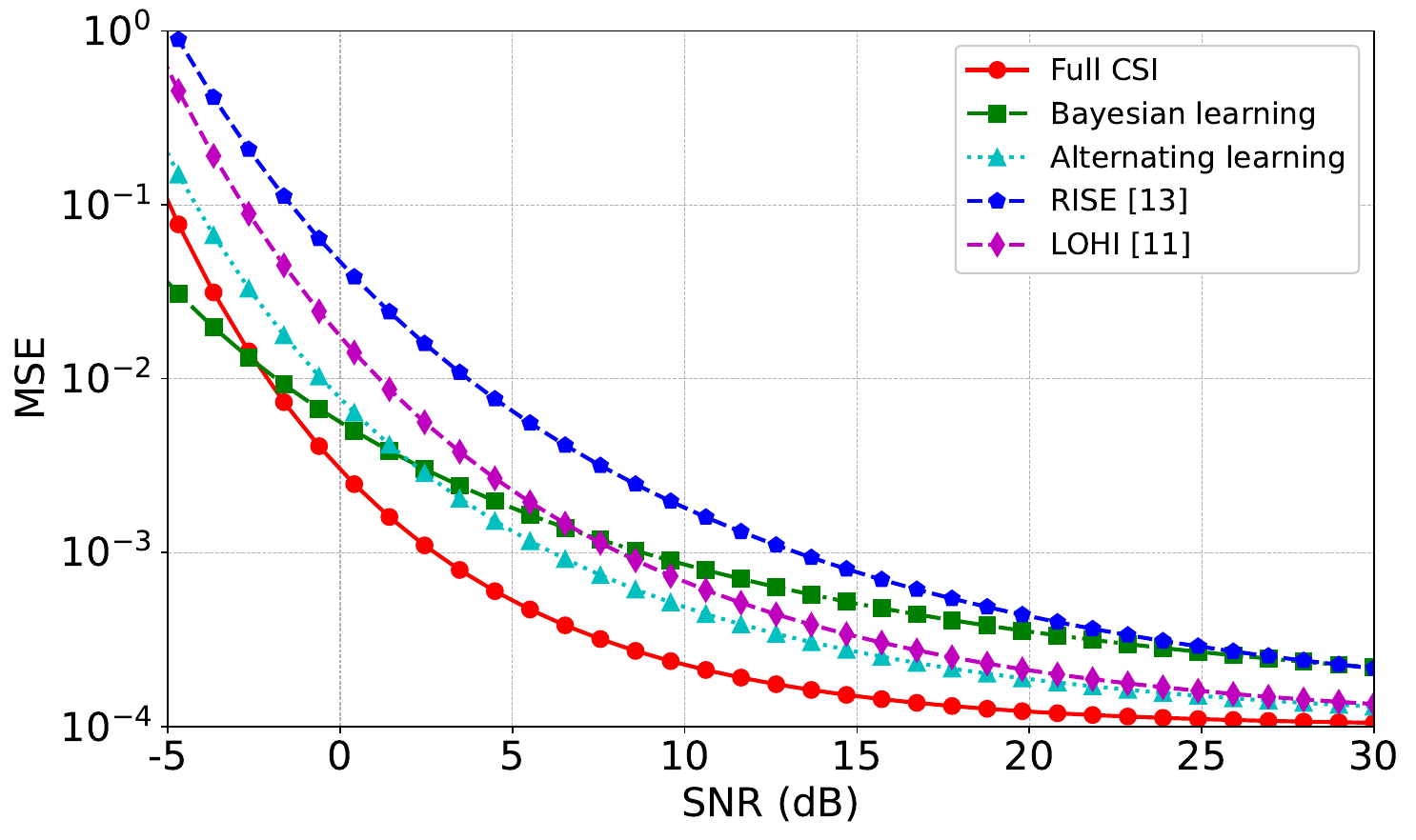}
    \caption{Sum MSE versus SNR for the proposed full CSI, alternating learning, and Bayesian learning schemes, compared with LOHI~\cite{yu2023learning} and RISE~\cite{huang2020reconfigurable} for $N=16$.}
    \label{fig:msesnr}
\end{figure}

Fig.~\ref{fig:msesnr} plots the sum MSE versus SNR for our alternating and Bayesian learning schemes with \(N=16\) IRS elements, alongside the full‐CSI lower bound and two baselines, LOHI~\cite{yu2023learning} and RISE~\cite{huang2020reconfigurable}. Both proposed methods deliver competitive MSE—particularly in the 5–10 dB range—while requiring far less CSI overhead. Although the full‐CSI case attains the minimum MSE, it is impractical in real deployments. Our schemes outperform LOHI and exceed RISE in efficiency without relying on large training datasets or complex action spaces, demonstrating their effectiveness in dynamic wireless environments with limited channel knowledge.

\begin{figure}[!t]
    \centering
    \includegraphics[scale=0.22]{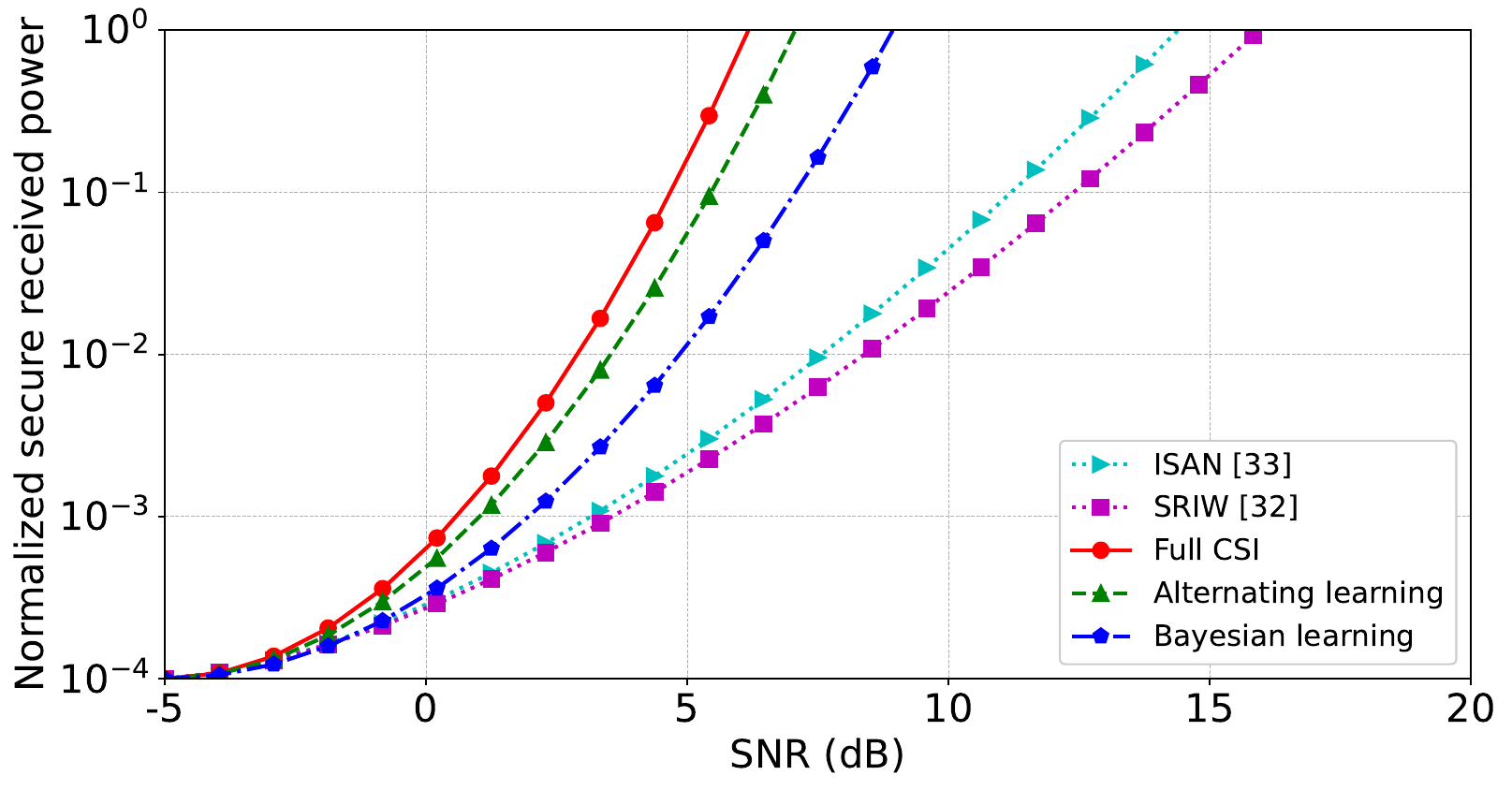}
    \caption{Secure WPT performance comparison in terms of the normalized secure harvested power ratio $\zeta$ versus SNR for $N=32$ IRS elements. Higher $\zeta$ indicates better secure power transfer efficiency. The proposed methods (Alternating learning and Bayesian learning) are compared against the full CSI benchmark and baseline schemes SRIW~\cite{Zhang2022} and ISAN~\cite{Xu2023}.}
    \label{fig:msenrefcomp}
\end{figure}

Fig.~\ref{fig:msenrefcomp} compares secure WPT performance via $\zeta = \frac{P_{\text{harv}}}{P_{\text{tx}}P_{\text{leak}}}$, where larger \(\zeta\) indicates more power harvested by honest users relative to eavesdroppers. The full‐CSI benchmark achieves the highest \(\zeta\), with Alternating learning closely approaching this optimum. Bayesian learning, despite lacking CSI, converges to near‐optimal \(\zeta\) at high SNR. In contrast, SRIW~\cite{Zhang2022} and ISAN~\cite{Xu2023} exhibit lower \(\zeta\) due to non‐adaptive designs, underscoring the advantage of our dynamic optimization framework.

\section{Conclusion}
\label{sec:conclusion}
This paper introduced a robust and secure opportunistic spectrum access framework using a Bayesian approach. Our distributed learning algorithm efficiently senses available spectrum in dynamic and hostile environments, feeding into a novel IRS-assisted downlink transmission scheme. Notably, the framework operates with partial or without CSI by leveraging alternating and Bayesian optimization to minimize MSE in high-dimensional settings. It addresses high-dimensional optimization challenges while supporting versatile applications such as multi-user MIMO and wireless power transfer, and ensuring fairness among users. Simulation results validate our solutions, providing a secure, adaptive, and efficient foundation for scalable, resilient, and intelligent 6G networks.

{\footnotesize
\bibliographystyle{IEEEtran}

}

\appendices

{\small
\setlength{\parskip}{0pt}
\setlength{\itemsep}{0pt}           
\setlength{\abovedisplayskip}{2pt plus 1pt minus 1pt}
\setlength{\belowdisplayskip}{2pt plus 1pt minus 1pt}
\setlength{\abovedisplayshortskip}{1pt}
\setlength{\belowdisplayshortskip}{1pt}

\section*{Appendix A: Proof of Distributed Algorithm Convergence and Accuracy}
\label{appendix:detection}

Consider a network comprising $K_H$ honest users indexed by $\mathcal{K}_H$, and $K_B$ Byzantine users indexed by $\mathcal{K}_B$. For each honest user $k \in \mathcal{K}_H$, define the log-belief
\begin{equation}
    \psi_k(t) = \ln\!\left( \frac{\pi_k(t)}{1 - \pi_k(t)} \right),
\end{equation}
which evolves according to
\begin{equation}
    \psi_k(t) = \psi_k(t{-}1) + \ell_k(t),
\end{equation}
where $\ell_k(t) = \ln L_k(\boldsymbol{y}_k(t))$ is the log-likelihood ratio of the local observation $\boldsymbol{y}_k(t)$. Byzantine users may arbitrarily manipulate their own decision values.

Assume the log-likelihood ratios $\{\ell_k(t)\}_t$ for honest users are i.i.d. with
\begin{equation}
    \mathbb{E}[\ell_k(t)] = D_{\mathrm{KL}}, 
    \qquad 
    \mathrm{Var}[\ell_k(t)] = \sigma_\ell^2 < \infty,
\end{equation}
where $D_{\mathrm{KL}} \equiv D_{\mathrm{KL}}\!\big(g_k(\cdot\mid\mathcal{H}_1)\,\|\,g_k(\cdot\mid\mathcal{H}_0)\big)$ is the Kullback–Leibler divergence between the true and null models. By a Chernoff-type concentration bound, for any $\epsilon>0$ and $k \in \mathcal{K}_H$,
\begin{equation}
    \mathbb{P}\!\left( \left| \frac{1}{t} \sum_{s=1}^{t} \ell_k(s) - D_{\mathrm{KL}} \right| \geq \epsilon \right) 
    \leq 2 \exp\!\left( -\frac{t \epsilon^2}{2 \sigma_\ell^2} \right).
\end{equation}

Each user updates its decision using the Byzantine-resilient trimmed-mean rule of Section~\ref{sec:detection}. Let $\mathcal{N}_k$ denote the neighborhood of user $k$ and assume $|\mathcal{N}_k| \ge 2K_B{+}1$. Form the trimmed set $\hat{\mathcal{N}}_k(t{-}1)$ by removing the $K_B$ largest and $K_B$ smallest values from $\{\delta_i(t{-}1)\}_{i\in\mathcal{N}_k}$. The update uses neighbors’ previous-round values:
\begin{equation}
    \delta_k(t) 
    = \min\!\left\{ \pi_k(t),\ \frac{1}{|\hat{\mathcal{N}}_k(t{-}1)|} \sum_{i \in \hat{\mathcal{N}}_k(t{-}1)} \delta_i(t{-}1) \right\}.
\end{equation}
If $|\mathcal{N}_k| < 2K_B{+}1$, we set $\delta_k(t)=\pi_k(t)$. Note that $|\hat{\mathcal{N}}_k(t{-}1)| \ge |\mathcal{N}_k| - 2K_B \ge 1$.

Let $\mathcal{G}_H$ be the subgraph induced by honest users, and define
\begin{align}
  \gamma 
  &= \min_{k\in\mathcal{K}_H}\bigl|\mathcal{N}_k\cap\mathcal{K}_H\bigr|,\\
  \pi_{k,\infty} 
  &= \lim_{t\to\infty}\pi_k(t)
    \quad\bigl(\text{the belief in the absence of adversaries}\bigr),\\
  e_k(t) 
  &= \bigl|\delta_k(t)-\pi_{k,\infty}\bigr|,\\
  e(t) 
  &= \max_{k\in\mathcal{K}_H}e_k(t).
\end{align}

\emph{Deviation of the trimmed average.} Let $H_k \triangleq \mathcal{N}_k\cap\mathcal{K}_H$ denote the honest neighbors of $k$, with $|H_k|\ge \gamma$. After trimming $K_B$ largest and $K_B$ smallest entries from $\{\delta_i(t{-}1)\}_{i\in\mathcal{N}_k}$, the trimmed set $\hat{\mathcal{N}}_k(t{-}1)$ satisfies $|\hat{\mathcal{N}}_k(t{-}1)| \ge |H_k| - 2K_B$. Even if some adversarial values remain in $\hat{\mathcal{N}}_k(t{-}1)$, their contribution is bounded because $\delta_i(\cdot)\in[0,1]$. Comparing the trimmed mean to the honest-neighbor mean yields
\begin{align}
\left|
\frac{1}{|\hat{\mathcal{N}}_k|}\!\!\sum_{i\in\hat{\mathcal{N}}_k}\!\delta_i(t{-}1)
-\frac{1}{|H_k|}\!\sum_{i\in H_k}\!\delta_i(t{-}1)
\right|
\;\le\; \frac{2K_B}{|H_k|}
\;\le\; \frac{2K_B}{\gamma}.
\label{eq:trim_dev}
\end{align}
The inequality follows from at most $2K_B$ honest values being removed by trimming and the bounded range of $\delta_i$.

Using \eqref{eq:trim_dev} together with concentration for the log-likelihood sums (which control $\pi_k$ via the logit), a union bound over $K_H$ honest users gives, for any $\epsilon > \frac{2K_B}{\gamma}$,
\begin{equation}
    \mathbb{P}\!\left( e(t) \ge \epsilon \right)
    \;\le\; 2 K_H \exp\!\left( - \frac{t \big(\epsilon - \tfrac{2K_B}{\gamma}\big)^2}{8 \sigma_\ell^2} \right).
\end{equation}
Equivalently, for any $\delta\in(0,1)$, with probability at least $1-\delta$,
\begin{equation}
    \max_{k\in\mathcal{K}_H} |\delta_k(t) - \pi_{k,\infty}| 
    \;\le\; \frac{2K_B}{\gamma} \;+\; \sqrt{\frac{8\sigma_\ell^2}{t}\,\log\!\frac{2K_H}{\delta}}.
\end{equation}

Under $\mathcal{H}_1$, $\psi_k(t)=\sum_{s=1}^t \ell_k(s)$ grows linearly with rate $D_{\mathrm{KL}}>0$, so $\pi_k(t)=\sig(\psi_k(t))\to 1$ exponentially. Since $\delta_k(t)\le \pi_k(t)$ by construction and the trimmed mean among neighbors concentrates near the honest mean, we obtain
\begin{equation}
    \limsup_{t\to\infty} \frac{1}{t}\,\ln\!\big(1-\delta_k(t)\big) \;\le\; -\,D_{\mathrm{KL}},
    \qquad \text{a.s.}
\end{equation}

Let $L_H$ be the (combinatorial) Laplacian of $\mathcal{G}_H$, and $\lambda_2(L_H)$ its algebraic connectivity. Standard mixing arguments for linear consensus with bounded disturbances (here, the trimmed-mean bias) yield
\begin{equation}
\begin{split}
\max_{k,j\in\mathcal{K}_H}\bigl|\delta_k(t) - \delta_j(t)\bigr|
&\;\le\;
C\,e^{-t\,\lambda_2(L_H)/2}
\\[-0.25ex]
&\quad+\;
\frac{2K_B}{\gamma}\,\frac{1}{1 - e^{-\lambda_2(L_H)/2}}\,,
\end{split}
\end{equation}
for some $C>0$ depending on the initial dispersion. In particular, when $K_B=0$ the additive term vanishes and exponential consensus holds; for $K_B>0$, the dispersion approaches an $O\!\big(\tfrac{K_B}{\gamma}\big)$ neighborhood of consensus, which shrinks as the honest degree increases.
\vspace{-2em}
\begin{flushright}
    $\blacksquare$
\end{flushright}

\section*{Appendix B: Proof of Convergence of Alternating Algorithm}
\label{app:convergence_proof}

We establish the convergence of the proposed alternating optimization algorithm. For brevity, denote $\mathcal{L}^{(t)} \triangleq \mathcal{L}(\mathcal{X}^{(t)}, \bm{\lambda}^{(t)})$, where $\mathcal{X}^{(t)} = (\bm{W}^{(t)}, \bm{\Phi}^{(t)}, \bm{C}^{(t)})$ are the primal variables. Let $\mathcal{F}_2 \equiv \mathcal{U}=\{\bm{\Phi}:\ |[\bm{\Phi}]_{n,n}|=1,\ \forall n\}$ denote the unit-modulus torus, and let $\mathcal{F}_1$ collect any simple convex bounds imposed on $(\bm{W},\bm{C})$ (the projection onto $\mathcal{F}_1$ can be the identity if no such bounds are used; power and leakage constraints are handled via the augmented Lagrangian terms and dual updates).

\subsection*{Monotonic Descent Property}

Each alternating step (exact or inexact via a projected gradient step) minimizes the augmented Lagrangian over one block while fixing the others:
\begin{IEEEeqnarray}{rCl}
&(\bm{W}^{(t+1)}, \bm{C}^{(t+1)}) \in \arg\min_{\bm{W},\,\bm{C} \in \mathcal{F}_1} \mathcal{L}(\bm{W}, \bm{\Phi}^{(t)}, \bm{C}, \bm{\lambda}^{(t)}), \IEEEyesnumber \label{eq:proof_wc_update} \\
&\bm{\Phi}^{(t+1)} \in \arg\min_{\bm{\Phi} \in \mathcal{F}_2} \mathcal{L}(\bm{W}^{(t+1)}, \bm{\Phi}, \bm{C}^{(t+1)}, \bm{\lambda}^{(t)}). \IEEEyesnumber \label{eq:proof_phi_update}
\end{IEEEeqnarray}
In implementation we use projected gradient steps with diminishing step sizes $\alpha_t,\beta_t>0$ satisfying Robbins--Monro conditions:
\begin{align}
\bm{W}^{(t+1)} &= \Pi_{\mathcal{F}_1} \!\left( \bm{W}^{(t)} - \alpha_t \nabla_{\bm{W}} \mathcal{L}(\bm{W}^{(t)}, \bm{\Phi}^{(t)}, \bm{C}^{(t)}, \bm{\lambda}^{(t)}) \right), \nonumber\\
\bm{C}^{(t+1)} &= \Pi_{\mathcal{F}_1} \!\left( \bm{C}^{(t)} - \alpha_t \nabla_{\bm{C}} \mathcal{L}(\bm{W}^{(t)}, \bm{\Phi}^{(t)}, \bm{C}^{(t)}, \bm{\lambda}^{(t)}) \right), \\
\bm{\Phi}^{(t+1)} &= \Pi_{\mathcal{F}_2} \!\left( \bm{\Phi}^{(t)} - \beta_t \nabla_{\bm{\Phi}} \mathcal{L}(\bm{W}^{(t+1)}, \bm{\Phi}^{(t)}, \bm{C}^{(t+1)}, \bm{\lambda}^{(t)}) \right).  \nonumber
\end{align}
Assuming Lipschitz-continuous gradients in each block and sufficiently small step sizes, the standard descent lemma yields
\begin{equation}
\begin{split}
\mathcal{L}^{(t+1)} 
&\le \mathcal{L}^{(t)}
- \tfrac{\eta_t}{2}\Bigl(
    \|\bm{W}^{(t+1)} - \bm{W}^{(t)}\|_F^2 
    + \|\bm{\Phi}^{(t+1)} - \bm{\Phi}^{(t)}\|_F^2 \\[-0.5ex]
&\quad\qquad\qquad\qquad
    + \|\bm{C}^{(t+1)} - \bm{C}^{(t)}\|_F^2
  \Bigr)\,.
\end{split}
\end{equation}

with $\eta_t = \min\{\alpha_t,\beta_t\}$, establishing monotonic decrease of the augmented Lagrangian.

\subsection*{Stationary Point Convergence}

Because $\mathcal{F}_2$ is compact and the iterates are bounded by construction (via the penalty and dual terms, and any simple bounds in $\mathcal{F}_1$), the sequence $\{\mathcal{X}^{(t)}\}$ admits limit points. Let $\mathcal{X}^*=(\bm{W}^*,\bm{\Phi}^*,\bm{C}^*)$ be a limit point of a convergent subsequence $\{\mathcal{X}^{(t_k)}\}$. First-order optimality of each block subproblem implies the variational inequalities
\begin{IEEEeqnarray}{rCl}
\langle \nabla_{\bm{W}} \mathcal{L}(\mathcal{X}^*, \bm{\lambda}^*),\, \bm{W} - \bm{W}^* \rangle &\geq& 0, \quad \forall \bm{W} \in \mathcal{F}_1, \IEEEyessubnumber \\
\langle \nabla_{\bm{\Phi}} \mathcal{L}(\mathcal{X}^*, \bm{\lambda}^*),\, \bm{\Phi} - \bm{\Phi}^* \rangle &\geq& 0, \quad \forall \bm{\Phi} \in \mathcal{F}_2, \IEEEyessubnumber \\
\langle \nabla_{\bm{C}} \mathcal{L}(\mathcal{X}^*, \bm{\lambda}^*),\, \bm{C} - \bm{C}^* \rangle &\geq& 0, \quad \forall \bm{C} \in \mathcal{F}_1, \IEEEyessubnumber
\end{IEEEeqnarray}
so $\mathcal{X}^*$ is a stationary point of the constrained augmented Lagrangian problem.

\subsection*{Dual Variable Convergence}

Dual variables are updated by projected gradient ascent using the most recent primal iterates:
\begin{IEEEeqnarray}{rCl}
\lambda_1^{(t+1)} &=& \left[ \lambda_1^{(t)} + \rho \left( \sum_{e=1}^{K_B} P_e^{\mathrm{(sig)}}(\mathcal{X}^{(t+1)}) - \Gamma_{\mathrm{leak}} \right) \right]_+, \IEEEyessubnumber \\
\lambda_2^{(t+1)} &=& \left[ \lambda_2^{(t)} + \rho \left( \|\bm{W}^{(t+1)}\|_F^2 - P_{\mathrm{max}} \right) \right]_+. \IEEEyessubnumber
\end{IEEEeqnarray}
Since $\mathcal{L}^{(t)}$ decreases and constraint violations are penalized quadratically, the violations diminish, and complementary slackness holds asymptotically:
\begin{equation}
\lim_{t \to \infty}
\left\| \begin{bmatrix}
\sum_{e=1}^{K_B} P_e^{\mathrm{(sig)}}(\mathcal{X}^{(t)}) - \Gamma_{\mathrm{leak}} \\
\|\bm{W}^{(t)}\|_F^2 - P_{\mathrm{max}}
\end{bmatrix} \right\| = 0.
\end{equation}
This establishes asymptotic consistency of the dual variables and completes the proof.
\qed

\section*{Appendix C: Proof of Convergence of Bayesian Algorithm}
\label{appendix:bo_proof}

Let $\mu_{t-1}(\boldsymbol{z})$ and $\sigma_{t-1}(\boldsymbol{z})$ denote the GP posterior mean and standard deviation at step $t{-}1$. By Theorem~2 of~\cite{lang2022geometric}, with probability at least $1-\delta$ the confidence bound holds uniformly over $\mathcal{Z}$:
\begin{equation}
\big|f(\boldsymbol{z}) - \mu_{t-1}(\boldsymbol{z})\big|
\le \sqrt{\beta_t}\, \sigma_{t-1}(\boldsymbol{z}), \qquad \forall\,\boldsymbol{z}\in\mathcal{Z}.
\end{equation}

Define the instantaneous regret $r_t := f(\boldsymbol{z}_t) - f(\boldsymbol{z}^*)$, where $\boldsymbol{z}^*$ is a best feasible design in the sense of Section~\ref{sec:bayesian_convergence}. Then
\begin{equation}
\begin{split}
r_t
&= \bigl(f(\boldsymbol{z}_t)-\mu_{t-1}(\boldsymbol{z}_t)\bigr)
   + \bigl(\mu_{t-1}(\boldsymbol{z}_t)-\mu_{t-1}(\boldsymbol{z}^*)\bigr)\\
&\quad+ \bigl(\mu_{t-1}(\boldsymbol{z}^*)-f(\boldsymbol{z}^*)\bigr)\\
&\le \sqrt{\beta_t}\,\sigma_{t-1}(\boldsymbol{z}_t)
    + \Delta_t
    + \sqrt{\beta_t}\,\sigma_{t-1}(\boldsymbol{z}^*)\,.
\end{split}
\end{equation}

where $\Delta_t := \mu_{t-1}(\boldsymbol{z}_t)-\mu_{t-1}(\boldsymbol{z}^*)$.

Under acquisition based on Expected Improvement with Constraints (EIC)—implemented by maximizing an improvement score multiplied by the posterior feasibility probability—one has the standard control
\begin{equation}
\Delta_t \le \sqrt{\beta_t}\,\big(\sigma_{t-1}(\boldsymbol{z}^*) - \sigma_{t-1}(\boldsymbol{z}_t)\big),
\end{equation}
which yields
\begin{equation}
r_t \le 2\sqrt{\beta_t}\,\sigma_{t-1}(\boldsymbol{z}_t).
\end{equation}

Summing over $t=1,\ldots,T$ and using Cauchy–Schwarz,
\begin{equation}
\begin{aligned}
R_T 
&= \sum_{t=1}^T r_t 
  \;\le\; 2 \sqrt{\beta_T}\sum_{t=1}^T \sigma_{t-1}(\boldsymbol{z}_t) \\[0.5ex]
&\le\; 2 \sqrt{\beta_T\,T}\Bigl(\sum_{t=1}^T \sigma_{t-1}^2(\boldsymbol{z}_t)\Bigr)^{\!1/2}
  \;\le\; \sqrt{C_1\,T\,\beta_T\,\gamma_T}\,.
\end{aligned}
\label{eq:regret_bound}
\end{equation}
where $C_1>0$ is a constant depending on the kernel and noise level, and $\gamma_T$ is the maximum information gain,
\begin{equation}
\gamma_T \;=\; \max_{\{\boldsymbol{z}_t\}_{t=1}^T}\;
\frac{1}{2}\log\det\!\left(\boldsymbol{I}_T + \frac{1}{\sigma^2}\,\boldsymbol{K}_T\right),
\end{equation}
with $\boldsymbol{K}_T$ the $T\times T$ kernel matrix $[\boldsymbol{K}_T]_{ij}=k(\boldsymbol{z}_i,\boldsymbol{z}_j)$.

Feasibility control follows from the probabilistic constraint handling in the acquisition. If the selection step restricts candidates to those satisfying a probability-of-feasibility filter
\begin{equation}
\mathbb{P}\big(g_j(\boldsymbol{z}_t)\le 0\big) \;\ge\; 1-\delta, \qquad \forall j,
\end{equation}
then the violation probability decays exponentially under sub-Gaussian noise and calibrated GP posteriors:
\begin{equation}
\mathbb{P}\!\left(g_j(\boldsymbol{z}_t) > 0\right) \;\le\; \exp\!\left(-\frac{C_3 t}{\sigma^2}\right),
\end{equation}
for some $C_3>0$, which implies via the Borel–Cantelli lemma~\cite{BeresnevichVelani2023} that constraint violations occur only finitely often almost surely. This completes the proof.
\qed}

\end{document}